\documentclass[10pt,journal]{IEEEtran}
\usepackage{amsfonts,amsthm,epsf,graphicx,amssymb,cite}
\usepackage{color}
	\usepackage{caption}
	\usepackage{amsmath}
\usepackage{listings}
\usepackage{multirow}
\usepackage{fancybox,mathtools}
\usepackage{tikz}
\usetikzlibrary{arrows,matrix,positioning}
\ifCLASSOPTIONcompsoc
 \usepackage[caption=false,font=normalsize,labelfont=sf,textfont=sf]{subfig}
\else
 \usepackage[caption=false,font=footnotesize]{subfig}

\newtheorem{lemma}{Lemma}

\newcommand{\tvec}[1]{\ensuremath{\Tilde{\vec{#1}}}}

\renewcommand{\vec}[1]{\ensuremath{\vec{#1}}}

\newcommand{\norm}[1]{\ensuremath{\| #1 \|}}
\newcommand{\mc}[1]{\ensuremath{\mathcal{#1}}}

\newcommand{\Complex}{{\mathbb{C}}}

\newcommand{\Nat}{{\mathbb{N}}}

\newcommand{\floor}[1]{\lfloor #1 \rfloor}

\newcommand{\ceil}[1]{\lceil #1 \rceil}

\DeclareMathOperator{\E}{E}

\DeclareMathOperator{\tr}{tr}

\DeclareMathOperator{\diag}{diag}

\renewcommand{\eqref}[1]{(\ref{eq:#1})}

\newcommand{\Figref}[1]{Figure~\ref{fig:#1}}

\newcommand{\secref}[1]{Section~\ref{sec:#1}}
\newcommand{\appref}[1]{Appendix~\ref{app:#1}}

\renewcommand{\vec}[1]{\ensuremath{\mathbf{#1}}}

\normalsize

\begin{document}

\title{Change Detection with Sparse Signals \\ using Quantum Designs}
\author{
\IEEEauthorblockN{Aditi Jain \IEEEauthorrefmark{1}, Pradeep Sarvepalli \IEEEauthorrefmark{2}, 
Srikrishna Bhashyam \IEEEauthorrefmark{3},  Arun Pachai Kannu \IEEEauthorrefmark{4}}\\
  \IEEEauthorblockA{Department of Electrical Engineering \\
		    Indian Institute of Technology Madras \\
		    Chennai - 600036, India\\
		    Email: \IEEEauthorrefmark{1}ee15s081@ee.iitm.ac.in, \IEEEauthorrefmark{2}pradeep@ee.iitm.ac.in,
			   \IEEEauthorrefmark{3}skrishna@ee.iitm.ac.in, \IEEEauthorrefmark{4}arunpachai@ee.iitm.ac.in }
}

\maketitle

\begin{abstract}
We consider the change detection problem where the pre-change observation vectors are purely noise and the post-change observation vectors 
are noise-corrupted compressive measurements of sparse signals with a common support, measured using a sensing matrix. 
In general, post-change distribution of the observations depends on parameters such as the support and variances of the sparse signal. 
When these parameters are unknown, we propose two approaches. In the first approach, we approximate the post-change pdf based on the known
parameters such as mutual coherence of the sensing matrix and bounds 
on the signal variances. In the second approach, we parameterize the post-change pdf with an unknown parameter and try to adaptively 
estimate this parameter using a stochastic gradient descent method. In both these approaches, we employ CUSUM algorithm with 
various decision statistics such as the energy of the observations, correlation values with columns of the sensing matrix and the 
maximum value of such correlations. We study the performance 
of these approaches and offer insights on the relevance of different decision statistics in different SNR regimes.  We also address the problem 
of designing sensing matrices with small coherence by using designs from quantum information theory. 
One such design, 
called SIC POVM, also has an additional structure which allows exact computation of the post-change pdfs of some  decision statistics even when the support set of the sparse signal is unknown. We apply our detection algorithms with SIC POVM based sequences to a massive random access problem and
show their superior performance over conventional Gold codes. 
\end{abstract}


\section{Introduction}
\label{sec:Intro}

The problem of change detection using statistical tests has been studied over several decades \cite{Lorden:AMS:71,Pollak:AS:85,Banerjee:tut}. 
The simplest model for change detection problems is described below.
The observation at time $t$ is denoted as $y[t]$. Let $\nu \in \Nat^+$ denote the change point such that 
the observations before and after change follow different statistics.
Specifically, the observations $\{y[t],~t \geq 0 \}$ 
are independent and follow the statistics,
\begin{eqnarray}
y[t] &\sim& \left\{ \begin{array}{ll} f_{0} & 0 \leq t < \nu, \\ f_{1} & t \geq \nu, \end{array} \right. \label{eq:stat}
\end{eqnarray}
where $f_0$ and $f_1$ are the pre-change and post-change probability density functions (pdf) respectively. 
When the change point $\nu$ is unknown and non-random, the quantities of interest are the average run length $T_r$ and the 
worst-case detection delay $D_w$. These quantities are mathematically defined below. 
We use $\E_{\nu}$ to denote the expectation with respect to the probability measure on the observations when the change point is $\nu$. We set $\nu=\infty$ when there is no change. With $T$ being the time at which a given algorithm declares change (which is random), the average run length and the worst case detection delay of the algorithm is given as
\begin{eqnarray}
T_r &=& \E_\infty \{ T \}, \label{eq:rl} \\
D_w &=& \sup_{\nu \geq 1} \E_\nu \{ (T-\nu) | T \geq \nu \}. \label{eq:wd} 
\end{eqnarray}
 The CUSUM algorithm \cite{Page:Biometrika:54} for change detection uses the log likelihood ratio (LLR) for each observation, which is given as $L(y[t]) = \log \frac{f_1(y[t])}{f_0(y[t])}$. The CUSUM metric $W[t]$ is initialized to $W[-1] = 0$ and is recursively computed as
\begin{equation*}
W[t] = \left (W[t-1] + L(y[t]) \right )^+,
\end{equation*}
where $(\cdot)^+$ denotes $\max \{\cdot,0\}$. The CUSUM decision rule $\mc{R}$ using the metric $W[t]$, with threshold $\tau \in (0,\infty)$,
is given as
\begin{eqnarray}
\mc{R} &=& \left \{\begin{array}{ll} \text{Declare change at time $t$} & \text{if}~ W[t] > \tau, \\ 
\text{Continue} & \text{otherwise.} \end{array} \right. \label{eq:cusum}
\end{eqnarray}
The threshold parameter $\tau$ in the above rule controls the average run length and the detection delay. 
It is shown in  \cite{Lorden:AMS:71} that CUSUM algorithm asymptotically minimizes the worst case detection delay, subject to 
a constraint on the average run length $T_r \geq \gamma$, as the threshold $\tau \rightarrow \infty$ 
(or equivalently as $\gamma \rightarrow 
\infty$).

Several variations of the model in \eqref{stat} have been addressed in the literature, considering cases where the pre-change  
\cite{mei} or post-change distributions \cite{Nikiforov:TIT:00} have unknown parameters. Adaptive algorithms to estimate 
the unknown parameters in the post-change distributions have been developed in \cite{Li:ICASSP:09,Singamasetty:ICASSP:17}. 

In natural and practical scenarios, most signals have sparse representations in an appropriately chosen basis. 
Compressive sensing   deals with the problem of reconstructing sparse signals from under-determined linear measurements \cite{Candes:SPM:08}. 
Orthogonal matching pursuit (OMP) is a popular sparse signal reconstruction technique \cite{Tropp:omp}
which works based on the correlation of the observation with columns of the sensing matrix. Detecting sparse signals 
in the presence of noise has been addressed in several papers such as \cite{Duarte:ICASSP:06,Davenport:JSTSP:10,Wimalajeewa:TSIPN:17, Haupt:ICASSP:07,Wang:ICASSP:08}, where detection is performed based on various statistics such as energy, correlation values and partially 
recovered support. In \cite{Zheng:ICC:11,Lei:ICSP:10}, the authors have developed a sequential approach based on LLR to detect sparse signals in the presence of noise.

In this paper, we consider the change detection problem wherein the pre-change observation vectors are purely noise and 
the post-change observation vectors are noise-corrupted compressive measurements of sparse signals with a common support, measured using
a sensing matrix. When the support and the variances of the non-zero entries of a sparse signal are unknown, the post change distributions of the observations (and other decision statistics) are not known perfectly. Change 
detection with sparse signals have been previously addressed in \cite{Alippi:ICASSP:14,Fellouris:ICASSP:17}. While \cite{Alippi:ICASSP:14} 
addresses the problem where the sparsifying dictionary of the signal is unknown, \cite{Fellouris:ICASSP:17} addresses the case where the 
observation has the same dimension as that of the sparse signal. In our work, the sparsifying dictionary is assumed to be known. However, we allow the dimension of the observation  to be much smaller than the dimension of the sparse signal. We develop change detection algorithms using various decision statistics and show their relevance in regimes with different signal to noise ratio (SNR). We also design sensing matrices using constructions from 
quantum information theory and show that they perform better than random constructions. More details on our system model and contributions 
are discussed in the following section.
\footnote{\emph{Notation:} Scalars are denoted by lowercase letters. Matrices (vectors) are denoted by uppercase (lowercase) boldface letters. The $i$-th column (entry) of $\vec{A}$ ($\vec{x}[t]$) is denoted by $\vec{a}_i$ ($x_i[t]$). The entry in $i$-th row and $k$-th column of $\vec{A}$ is denoted by $a_{ik}$. We denote transpose by $(\cdot)^T$, conjugate transpose by $(\cdot)^*$, inverse by $(\cdot)^{-1}$, trace by $tr(\cdot)$, $\ell_p$ norm by $\norm{\cdot}_p$. Calligraphic letters denote sets like $\mathcal{S}$. We use $\vec{A}_\mathcal{S}$  ($\vec{x}_\mathcal{S}$) to denote the sub-matrix (sub-vector) of $\vec{A}$  ($\vec{x}$) consisting of columns (entries) whose index belongs to $\mathcal{S}$. $|\cdot |$ denotes the absolute value of a scalar, as well as the cardinality of a set, which will be apparent from the context. $\floor{.}$ and $\ceil{.}$ denote the floor and ceil values of their arguments. We use $\tilde{\theta}$ to denote an approximation of $\theta$. We denote a zero vector of dimension $N$ by $\vec{0}_N$, identity matrix by $\vec{I}$ and $\sqrt{-1}$ by $j$. We use $\diag([\vec{d}])$ to denote a diagonal matrix with elements of vector $\vec{d}$ as its diagonal entries. For observations $d[t]$, we denote pre-change by $f_0^D$ and post-change pdf by $f_1^D$. $\tilde{f}_1^D$ denotes the approximation of $f_1^D$ and $f_1^{D,\theta}$ denotes the post-change pdf of $d[t]$ parametrized by $\theta$. $\mathrm{N}(\mu,\sigma^2)$ denotes Gaussian and $\mathrm{CN}(\mu,\sigma^2)$ denotes complex Gaussian distribution with mean $\mu$ and variance $\sigma^2$. $\exp{(\lambda)}$ denotes exponential distribution with rate parameter $\lambda$. $\mathcal{X}^2_{k}$ denotes central and $\mathcal{X}^2_{k}(\mu)$ denotes non central chi squared distribution with $k$ degrees of freedom and non centrality parameter $\mu$.}

\section{System Model} 
\label{sec:smod}
\subsection{Sparse Signal Model}
For the change detection problem with sparse signals, we consider the vector observation model,
\begin{eqnarray}
\vec{y}[t] &=& \left\{ \begin{array}{ll} \vec{n}[t] & 0 \leq t < \nu, \\ \vec{A} \vec{x}[t] + \vec{n}[t] & t \geq \nu. \end{array} \right. \label{eq:model}
\end{eqnarray}
Here, $\vec{n}[t] \in \Complex^{M \times 1}$ denotes the complex additive white Gaussian noise (AWGN) with pdf $\mathrm{CN}(\vec{0},\sigma_n^2 \vec{I})$, $\vec{A} \in \Complex^{M \times N}$ denotes the sensing matrix (with $M \leq N$) and $\vec{x}[t]$ denotes the sparse signal with the number of non-zero entries $\norm{\vec{x}[t]}_0 = K \ll N$. We refer to $K$ as the sparsity level of the signal $\vec{x}[t]$. We consider the case where the support (i.e. locations of non-zero entries) of $\vec{x}[t]$ remains the same for all 
$t \geq \nu$. Let $\mc{S}$ denote the ordered support set, containing the locations of non-zero entries of $\vec{x}[t]$. Note that  $\vec{x}_\mc{S}[t]$ is of size $K$ and contains the non-zero entries of $\vec{x}[t]$. After the change, $\vec{y}[t]$ can be restated as
\begin{eqnarray}
\label{eq:y_in_terms_supp}
\vec{y}[t]=\vec{A}_\mc{S}\vec{x}_\mc{S}[t]+\vec{n}[t]= \sum_{i \in \mc{S}} \vec{a}_ix_i[t]+\vec{n}[t], \hspace{2mm} \forall t\geq \nu.
\end{eqnarray} 
We are interested in detecting the change and finding the support of the sparse signal $\vec{x}[t]$, once the change is declared.  
To proceed further, we assume that the pdf of $\vec{x}_\mc{S}[t]$ is $\mathrm{CN}(\vec{0},\vec{C_x})$ and the signal  covariance matrix 
$\vec{C_x} = \diag([\sigma_1^2,\cdots,\sigma_K^2]))$ is diagonal. 
We also assume that the non-zero entries of $\vec{x}[t]$ are independent across time $t$. 
For the model in \eqref{model}, we assume that the noise variance $\sigma_n^2$ is known. Mutual coherence of the sensing matrix, defined as
\begin{equation}
\alpha = \max_{1\leq k \neq \ell \leq N} \frac{|\langle \vec{a}_k,\vec{a}_\ell \rangle|}{\norm{\vec{a}_k} \norm{\vec{a}_\ell}}, \label{eq:alp}
\end{equation}
plays an important role in the performance of sparse signal recovery algorithms \cite{Tropp:omp}. In general, 
smaller the value of $\alpha$, better will be the sparse signal recovery performance. 

\subsubsection{Change Detection Algorithms}
For the special case when the sensing matrix $\vec{A} = \vec{I}$ in \eqref{model}, the change detection problem is addressed in \cite{Fellouris:ICASSP:17}. On the other hand, we consider the general case, which allows \emph{compressive} measurements ($M \ll N$) on the 
sparse signal. The change detection algorithms and their performance vary greatly depending on whether the three parameters, namely,  
support set $\mc{S}$, sparsity level $K$ and signal covariance $\vec{C_x}$, are known or unknown. We consider all the combinations regarding the knowledge of these three parameters and develop corresponding change detection algorithms. When the signal variance is unknown, we assume that
the lower $\sigma_{\min}^2$ and the upper $\sigma_{\max}^2$ bounds are available such that $\sigma_{\min}^2 \leq \sigma_i^2 \leq 
\sigma_{\max}^2$, for $1 \leq i \leq K$. For change detection, we use CUSUM algorithm with different decision statistics such as 
energy $\|\vec{y}[t]\|_2^2$, correlation values $\vec{g}[t] = \vec{A}^* \vec{y}[t]$ and the maximum of these correlations $\|\vec{g}[t]\|_\infty^2$. 
When any or all of the three parameters $\mc{S}$, $K$ and $\vec{C_x}$ are unknown, we have the following two approaches:
\begin{itemize}
\item In the \emph{pdf-approximation} based approach, we approximate the post-change pdf based on the known parameters such as $\alpha$, 
$\sigma_{\min}^2$ and $\sigma_{\max}^2$ and use this approximate pdf for LLR computations. Here, we use the philosophy that the change detection
will be most difficult when the post-change pdf is closest (in terms of Kullback-Leibler (KL) distance)  to the pre-change pdf and hence try to obtain 
the \emph{worst-case} post-change pdf. We also ensure that some of our post-change pdf approximations are exact when the sensing matrix has some additional structure (such as when it is unitary). 
\item In the \emph{parameter-estimation} based approach, we parameterize the post-change pdf with an unknown parameter and try to adaptively 
estimate this parameter using a stochastic gradient descent method \cite{kiefer,Singamasetty:ICASSP:17}. 
LLR computations are done using the estimated parameter value in the parameterized post-change pdf.
\end{itemize}
The various change detection algorithms are presented in \secref{cdet}.
 
\subsubsection{Sensing Matrix Design}
From the theory of compressive sensing \cite{Candes:SPM:08}, sensing matrices with small mutual coherence are better suited for sparse signal
recovery \cite{Li:TIT:14}. 
Towards this, we design sensing matrices with small coherence using designs from quantum information theory \cite{Klappenecker:ISIT:05}.
Specifically, we use 
symmetric informationally complete positive operator valued measure (SIC POVM) from the theory of equi-angular lines 
\cite{Scott:JMP:10,Grassl:JMP:17}, mutually unbiased bases (MUB) \cite{Klappenecker:FQ7:04} and approximate MUBs \cite{Shparlinski:TI:06} 
from quantum information theory. In addition to having low mutual coherence, we also show that SIC POVM based sensing matrix
has an additional structure, using which,  exact computation of the post-change pdfs of some decision statistics is possible, even when the 
support set $\mc{S}$ of the sparse signal is unknown.  The sensing matrix design problem is addressed in \secref{sdes}. 

\subsection{Applications of the Model}
In this section, we discuss some of the applications of our change detection model in \eqref{model}.

\subsubsection{Random Access in Direct Sequence-Code Division Multiple Access (DS-CDMA)}

Consider the synchronous DS-CDMA system with a codebook $\{\vec{b}_i : 1\leq i\leq N\}$, where the  codes $\vec{b}_i$ are $M\times 1$ vectors. 
Suppose $Q$ users indexed by a \emph{known} set $\mc{Q}$ are currently active and  at 
time $t = \nu$, $K$ new users indexed by an \emph{unknown} set $\mc{S}$ become active. 
The corresponding observation model is,
\begin{eqnarray}
\vec{r}[t] &=& \left\{ \begin{array}{ll} \vec{B}_\mc{Q} \vec{x}_\mc{Q}[t] + \vec{n}[t] & 0 \leq t < \nu, \\ 
\vec{B}_\mc{Q} \vec{x}_\mc{Q}[t]+\vec{B}_\mc{S} \vec{x}_\mc{S}[t] + \vec{n}[t] & t \geq \nu. \end{array} \right. \label{eq:dsmodel}
\end{eqnarray}
Each entry in $\vec{x}_\mc{Q}$ (and $\vec{x}_\mc{S}$) is the product of the (flat) fading channel gain and the constellation symbol sent by the corresponding user in the index sets $\mc{Q}$ (and $\mc{S}$). The goal is to detect the change and identify the 
new users entering into the system. The authors in \cite{Oskiper:TIT:02} consider the above model \eqref{dsmodel} for the special case of $K=1$. We allow $K>1$, that is, more than one user can enter the system at a given time $\nu$.  

If the information related to already active users $\vec{x}_\mc{Q}[t]$ is known (from detection/estimation), then it can be simply subtracted out from the received signal
as $\vec{y}[t] = \vec{r}[t] - \vec{B}_\mc{Q} \vec{x}_\mc{Q}[t]$, resulting in the model \eqref{model}. On the other hand, if $\vec{x}_\mc{Q}[t]$ is not perfectly known, we can project the observation $\vec{r}[t]$ onto the orthogonal complement of $\vec{B}_\mc{Q}$ as $\vec{y}[t]= \vec{P}_\mc{Q}^{\perp} \vec{r}[t]$ with $\vec{P}_\mc{Q}^{\perp}=\vec{I}- \vec{B}_\mc{Q}(\vec{B}^*_\mc{Q} 
\vec{B}_\mc{Q})^{-1}\vec{B}_\mc{Q}^*$. With this projection, the \emph{effective} sensing matrix becomes, $\vec{A}_\mc{S} = \vec{P}_\mc{Q}^{\perp} \vec{B}_\mc{S}$, resulting in the model specified in \eqref{model}.

\subsubsection{Localized Change Detection in Sensor Networks} Consider a wireless sensor network which has $N$ sensor nodes and a fusion center.
To convey their identity to the fusion center and enable  transmission at the same time, each of the $N$ sensors is 
assigned a unique $M$-length code, $\lbrace \vec{a}_i: i=1,\dots,N \rbrace$.
All the sensors are initially in the OFF state, so that the observation at the fusion center is purely noise. 
When a change / event occurs at time $t = \nu$, a subset of sensors $\mc{S}$ get affected 
by that event and enter into the ON state.  The sensors in the ON state send their information symbol multiplexed with their corresponding code. 
The observation at the fusion center after the change is $\vec{y}[t]=\sum_{i\in \mc{S}} \vec{a}_i x_i[t] + \vec{n}[t]$, with
$x_i[t]$ being the product of the channel between the $i$-th sensor to the fusion center and the information symbol 
sent by that sensor at time $t$. This resembles the model in \eqref{y_in_terms_supp}. Here, recovering the support set $\mc{S}$ 
reveals the identities of the affected sensors, which in turn can reveal information on the location of the event in the network. 

\subsubsection{User Activity Detection in Massive Random Access} Massive random access systems \cite{Boljanović:SPAWC:17} with applications
in Internet of things (IoT), consist of a single receiving station and $N$ number of users, with $N$ being quite large. 
Each user is assigned an $M$-length identification code $\lbrace \vec{a}_i: 1 \leq i \leq N \rbrace$, which is known to the receiver. 
Initially, there is no active transmission. At some point in time, a small group of users indexed by the set $\mc{S}$ become active and
send their transmission using the codes assigned to them, resulting in an observation model specified in \eqref{y_in_terms_supp}.

\subsection{Main Contributions}

Some of the main contributions of our work are:

1) We develop change detection algorithms using the compressive measurements on the sparse signal and compare their performance in 
terms of worst-case detection delay versus average run length.

2) We develop an \emph{aggregate} CUSUM algorithm using the entire correlation vector $\vec{A}^* \vec{y}[t]$, which performs better 
than energy $\|\vec{y}[t]\|_2^2$ based detection and maximum correlation $\|\vec{A}^* \vec{y}[t]\|_\infty^2$ based detection, in most of the scenarios.

3) We show that energy $\|\vec{y}[t]\|_2^2$  based detection works better than maximum correlation $\|\vec{A}^* \vec{y}[t]\|_\infty^2$ based detection in the low SNR regime. On the other hand, when SNR is high, we show that correlation based detection performs better than energy based detection.

4) We show that quantum information theory based (deterministic) sensing matrices perform better than randomly generated matrices with i.i.d. Gaussian or 
Bernoulli distributed entries. Among the deterministic matrices, we show that SIC POVM yields the best detection performance when compared to MUB and approximate MUB based constructions.

5) We consider an application of our algorithms in massive random access and show that SIC POVM based codes have better detection performance when compared to
Gold codes.

\section{Change Detection Algorithms} \label{sec:cdet}
In this section, we develop change detection algorithms for the model considered in \eqref{model}. The statistics of the post-change observations
$\{\vec{y}[t],~t\geq \nu\}$ depend on the parameters such as support set $\mc{S}$, sparsity level $K$ and the signal covariance matrix $\vec{C_x}$. The change detection mechanism depends on whether these parameters are known or unknown. 
We address all the possible cases in this section.

\subsection{Both Support and Signal Variance Known}
For this case, both support set $\mc{S}$ and $\vec{C_x}$ are assumed to be known. We always have the pre-change pdf  of $\vec{y}[t]$  as 
$f_0^Y = \mathrm{CN}(\vec{0},\sigma_n^2 \vec{I})$.  When the support is known, the pdf of $\vec{y}[t]$  after the change is also perfectly known,
\begin{eqnarray}
\label{eq:Y_post}
f_1^Y &=& \mathrm{CN}(\vec{0}, \vec{A}_\mc{S} \vec{C_x} \vec{A}_\mc{S}^* +  \sigma_n^2 \vec{I}).
\end{eqnarray}

This resembles the standard change detection problem, for which the CUSUM algorithm is asymptotically optimal. It is described below 
for completeness. Computing the LLR as $L^Y(\vec{y}[t]) = \log \frac{f_1^Y(\vec{y}[t])}{f_0^Y(\vec{y}[t])}$, the CUSUM metric at each time $t$ is
\begin{eqnarray*}
W^Y[t] = \left (W^Y[t-1] + L^Y(\vec{y}[t]) \right )^+,
\end{eqnarray*}
with the initialization $W^Y[-1] = 0$. We use the metric $W^Y[t]$ to detect change based on the rule specified in \eqref{cusum}.
We refer to this method as \emph{Ideal}-CUSUM since the support is known perfectly in advance.

In some situations, we find it useful to implement the change detection algorithms using the vector of inner products $\vec{g}[t]$, which is defined as
\begin{equation}
\vec{g}[t] \triangleq \vec{A}^* \vec{y}[t]. \label{eq:gvec}
\end{equation}
We note that, when $\vec{A}$ is full rank, $\vec{g}[t]$ serves as a sufficient statistic for detection since $\vec{y}[t]$ can be obtained from $\vec{g}[t]$ as $\vec{y}[t] = (\vec{A} \vec{A}^*)^{-1} \vec{A} \vec{g}[t]$. The pre-change and post-change pdf of $\vec{g}[t]$ is given as
\begin{equation*}
\begin{split}
f_0^G &= \mathrm{CN}(\vec{0},\sigma_n^2 \vec{A}^*\vec{A}), \\
f_1^G &= \mathrm{CN}(\vec{0},  \vec{A}^* \vec{A}_\mc{S} \vec{C_x} \vec{A}_\mc{S}^* \vec{A} +  \sigma_n^2 \vec{A}^*\vec{A}). 
\end{split}
\end{equation*}
We compute the LLR as $L^G(\vec{g}[t]) = \log \frac{f_1^G(\vec{g}[t])}{f_0^G(\vec{g}[t])}$ 
and the CUSUM metric as $W^G[t] = \left (W^G[t-1] + L^G(\vec{g}[t]) \right )^+$. We can implement the CUSUM rule as given in \eqref{cusum}.

Consider the special case when sensing matrix $\vec{A}$ is unitary. In this case, 
\begin{eqnarray}
\label{eq:unitg}
\vec{g}[t] &=& \left\{ \begin{array}{ll} \tvec{n}[t] & 0 \leq t < \nu, \\ \vec{x}[t] + \tvec{n}[t] & t \geq \nu, \end{array} \right. 
\end{eqnarray}
where $\tvec{n}[t] = \vec{A}^* \vec{n}[t]$ and $\tvec{n}[t] \sim \mathrm{CN}(\vec{0},\sigma_n^2 \vec{I})$. Notice that whenever the index 
$i \notin \mc{S}$, the $i^{th}$ entry $g_i[t]$ in $\vec{g}[t]$,  is purely noise before and after the change point. Thus, we have $f_0^{G_i} = \mathrm{CN}(0,\sigma_n^2)$ and  $f_1^{G_i} = \mathrm{CN}(0,\sigma_n^2+\sigma_i^2)$, whenever $i \in \mc{S}$. Here, $\sigma_i^2$ denotes the variance of the $i^{th}$ entry in $\vec{x}_\mc{S}[t]$. Hence, the LLR $L^G[t]$ simplifies to, 
\begin{equation}
L^G(\vec{g}[t]) = \sum_{i \in \mc{S}} \log \frac{f_1^{G_i}(g_i[t])}{f_0^{G_i}(g_i[t])}. \label{eq:llr_sum}
\end{equation}
Hence, for a unitary sensing matrix, the LLR computation using $\vec{g}[t]$ boils down to summing the individual LLRs of only those entries in $\vec{g}[t]$ which belong to the support set $\mc{S}$. This observation will be useful in designing algorithms when the support set is unknown.

\subsection{Both Signal Variance and Sparsity Level Known, Support Unknown}
\label{sec:Unknown_Support_Case}
For the model in \eqref{model}, we consider the case where the support set of $\vec{x}[t]$ is unknown. However, we 
assume that the sparsity level $K$ and the signal covariance matrix $\vec{C_x}$ are known. In addition, we assume that the 
covariance matrix is of the form $\vec{C_x} = \sigma_x^2 \vec{I}$, i.e., all the non-zero entries are i.i.d. 
Under these assumptions, we develop the (asymptotically) optimal CUSUM and other sub-optimal detection techniques in this section.

\subsubsection{Optimal CUSUM}
When the support $\mc{S}$ of $\vec{x}[t]$ is unknown, the covariance matrix  
of the post-change observations in \eqref{Y_post} is unknown. Thus, we can treat the support $\mc{S}$ as an unknown parameter 
in the post-change pdf. Since there are only finite number of possibilities for the support, which is $\binom{N}{K}$, 
we run the CUSUM algorithm  simultaneously for each possible support. The exact details of the algorithm are given below. 

Let $\mathfrak{S}$ denote the set of all the subsets of $\{1,2,\cdots,N\}$ which have cardinality equal to $K$. Thus, the true support $\mc{S}$ is also one of the 
entries in $\mathfrak{S}$. Consider a candidate entry $\mc{\hat{S}}$ in $\mathfrak{S}$. 
For the candidate support set $\mc{\hat{S}}$, the associated post-change pdf of the observation is $f_1^{Y,\mc{\hat{S}}} = \mathrm{CN}(\vec{0},\sigma_x^2 \vec{A}_\mc{\hat{S}} \vec{A}_\mc{\hat{S}}^* + \sigma_n^2 \vec{I})$. Thus, we compute the LLR for each of the candidate support  $\mc{\hat{S}}$ 
as
\begin{equation*}
L^{Y,\mc{\hat{S}}}(\vec{y}[t]) = \log \frac{f_1^{Y,\mc{\hat{S}}}(\vec{y}[t])}{f_0^Y(\vec{y}[t])}.
\end{equation*}
It can be easily verified that, for $t \geq \nu$, we have 
$\E_\nu\{L^{Y,\mc{{S}}}(\vec{y}[t])\} > \E_\nu\{L^{Y,\mc{\hat{S}}}(\vec{y}[t])\}$, for any $\mc{\hat{S}} \in \mathfrak{S} \setminus \mc{S}$.
Hence, after the change, the expected value of LLR will be the highest when the candidate support set $\hat{\mc{S}}$ is identical to the true support $\mc{S}$. 
Since, we do not know the true support, we run CUSUM for each of the candidate support set and make a decision based on the CUSUM metric which has the largest magnitude. Specifically, for each $\mc{\hat{S}} \in \mathfrak{S}$, we compute the CUSUM metric as
\begin{equation*}
W^{Y,\mc{\hat{S}}}[t] = \left (W^{Y,\mc{\hat{S}}}[t-1] + L^{Y,\mc{\hat{S}}}(\vec{y}[t]) \right )^+,
\end{equation*} 
with the initialization $W^{Y,\mc{\hat{S}}}[-1] = 0$. The CUSUM change detection rule $\mc{R}$ is given as
\begin{equation}
\mc{R} = \left \{\begin{array}{ll} \text{Declare change at time $t$} & \text{if}~ \underset{\mc{\hat{S}} \in \mathfrak{S}} \max \hspace{2mm} W^{Y,\mc{\hat{S}}}[t] > \tau, \\ 
\text{Continue} & \text{otherwise.} \end{array} \right. \label{eq:ocusum}
\end{equation}
In \cite{Lorden:AMS:71}, Lorden also considered the case when the post-change pdf could be any one from a finite set of pdfs. He established
that running CUSUM separately for each possible post-change pdf and making the decision based on the maximum of these CUSUM metrics, as done in \eqref{ocusum}, is asymptotically optimal, as the average run length constraint approaches infinity. Though optimal, for large values of $N$, running $\binom{N}{K}$ separate CUSUMs can be prohibitively complex.

We now develop some sub-optimal detection techniques, which are described below.

\subsubsection{Aggregate CUSUM}
\label{sec:G}
We develop an algorithm, which we refer to as Aggregate CUSUM, which uses the vector of inner-products $\vec{g}[t]$, defined in \eqref{gvec}. 
To get some insight, we start by considering the case when $\vec{A}$ is unitary. In that case, $\vec{g}[t]$ is given by  \eqref{unitg} and the LLR computation \eqref{llr_sum} in Ideal-CUSUM  is equivalent to summing the LLRs of individual entries of $\vec{g}[t]$ corresponding to the non-zero locations. Since we do not know the support, we compute LLR for each entry $g_i[t]$ in $\vec{g}[t]$, assuming that $i$ belongs to the support  $\mc{S}$. Specifically, with $f_0^{G_i} = \mathrm{CN}(0,\sigma_n^2)$ and $f_1^{G_i} =  \mathrm{CN}(0, \sigma_n^2+\sigma_x^2)$, we compute the LLR for $i^{th}$ entry as
\begin{eqnarray}
 \label{eq:llr_par}
L^{G_i}(g_i[t]) = \log \frac{f_1^{G_i}(g_i[t])}{f_0^{G_i}(g_i[t])}.
\end{eqnarray} 
We compute CUSUM metric for each entry $g_i[t]$ parallelly as 
\begin{equation*}
W^{G_i}[t] = \left (W^{G_i}[t-1] + L^{G_i}(g_i[t]) \right )^+,
\end{equation*} 
with $W^{G_i}[-1]=0$.
Again, we can easily verify that, after the change ($t \geq \nu$), for any $i \in \mc{S}$ and any $\ell \notin \mc{S}$, we have
$\E_\nu\{L^{G_i}[t]\} > 0$ and $ \E_\nu\{L^{G_\ell}[t]\} < 0$. Hence, after the change, LLR in \eqref{llr_par} tends to be larger when the entry belongs to the true support. This implies that the CUSUM metrics corresponding to the non-zero locations tend to be higher after the change. This indirectly provides a way of identifying the unknown support $\mc{S}$ of $\vec{x}[t]$. In order to detect the change, we sum the $K$-largest CUSUM metrics at each time $t$ and compare it with a threshold. We declare change at time $t$ based on the following decision rule, 
\begin{eqnarray}
\mc{R} = \left \{\begin{array}{ll} \text{Declare change} & \text{if}~  \hspace{2 mm} \sum\limits_{\substack{i=0}}^{K-1} W^{G_i}_{(N-i)}[t] > \tau, \\ 
\text{Continue} & \text{otherwise,} \end{array} \right. \label{eq:pcusum}
\end{eqnarray}  
where $W^{G_i}_{(m)}[t]$ denotes the $m^{th}$ ordered statistic of the set  $\lbrace W^{G_i}[t]:1\leq i \leq N \rbrace$, with $W^{G_i}_{(N)}[t]$
being the largest. Using Aggregate CUSUM, we can also estimate the support by picking those  $K$ locations that correspond to the $K$- largest values of $W^{G_i}[T]$ where $T$ is the time at which change is declared. When $\vec{A}=\vec{I}$, the Aggregate CUSUM algorithm described 
above has been studied in \cite{Fellouris:ICASSP:17} and it was shown to have asymptotic optimality properties under some specific conditions. 

We extend the Aggregate CUSUM algorithm for a non-unitary sensing matrix $\vec{A}$ as follows. In this case, the  
the post-change pdf for each entry $f_1^{G_i}$ is not known perfectly and hence we use approximations for the post-change pdf. 
We need these approximations to be near-exact in order to 
ensure good detection performance. Towards getting the approximate pdf, we use the mutual coherence $\alpha$ of the sensing matrix $\vec{A}$ 
defined in \eqref{alp}. Based on the derivation in \appref{G_sigvar}, we set the approximate post-change pdf as
\begin{eqnarray}
\label{eq:G_post}
\tilde{f}_1^{G_i} =   \begin{cases}
    \mathrm{CN}(0,\sigma_n^2+ K\alpha^2 \sigma_x^2), & i \notin \mc{S}, \\
    \mathrm{CN}(0,\sigma_n^2 + K\alpha^2 \sigma_x^2 + (1-\alpha^2)\sigma_x^2), & i \in \mc{S}.
  \end{cases}
\end{eqnarray} 
Once we get the approximating pdf $\tilde{f}_1^{G_i}$, we proceed in the same way as before, by replacing $f_1^{G_i}$ with $\tilde{f}_1^{G_i}$ for ${i \in \mc{S}}$ in \eqref{llr_par}.

\subsubsection{Energy CUSUM}
\label{sec:E}
In this section, we describe another suboptimal technique, which uses energy of the received signal as the decision statistic. 
Energy detector has been used previously to detect sparse signals in \cite{Haupt:ICASSP:07,Wang:ICASSP:08}.
Let us define the energy of the observation vector to be $e[t] = \norm{\vec{y}[t]}_2^2$. Before the change, $e[t]$ is the sum of squares of $2M$ i.i.d zero mean Gaussian random variables of variance $\sigma_n^2$ and follows $\chi^2$ distribution with $2M$ degrees of freedom. This implies $f_0^E=\chi^2_{2M}$. However, for sufficiently large values of $M$, the $\chi^2$ distribution of $e[t]$ can be approximated as Gaussian with 
appropriate mean and variance, using the central limit theorem (CLT). The approximate pre-change distribution is $\tilde{f}_0^E = \mathrm{N}(M \sigma_n^2, M \sigma_n^4)$. Based on the derivation in \appref{E_sigvar}, we also obtain the approximate post-change pdf as
\begin{equation}
\label{eq:E_post}
\begin{split}
\tilde{f}_1^E =& \mathrm{N}(\mu_{E},\sigma^2_{E}) \hspace{2mm}\text{where,} \hspace{1mm}\mu_{E} = K\sigma_{x}^2+ M\sigma_n^2 \hspace{4mm} \text{and}\\
\sigma^2_{E}=& K\phi_{\min}^2 +2 \sigma_n^2  K\phi_{\min} +M(\sigma_n^2)^2.
\end{split}
\end{equation}
Here, $\phi_{\min}=\max \big\{0,\sigma_{x}^2 ( 1 - (K-1) \alpha) \big)\}$. Using these pdf approximations, we run the CUSUM algorithm for energy function $e[t]$. Note that, this method does not use the fact that support of all the signals $\vec{x}[t]$ remains same after the change.

\subsubsection{Correlator CUSUM}
\label{sec:C}
We now describe the matched-filter/ correlator based metric as the decision statistic. Specifically, considering the 
vector of inner products $\vec{g}[t] = \vec{A}^* \vec{y}[t]$, we use the maximum inner product (correlation value) 
$c[t] = \norm{\vec{g}[t]}^2_\infty$ as the decision statistic. The correlator based statistic has been previously used for detection of sparse signals in \cite{Duarte:ICASSP:06}. 
The pre-change pdf is the distribution of maximum of $N$ i.i.d. exponential random variables, 
$f_0^C = N (1-e^{-\lambda_n c[t]})^{N-1}\lambda_n e^{-\lambda_n c[t]}$, where $\lambda_n=\frac{1}{\sigma_n^2}$.
Based on the derivations in \appref{G_sigvar}, we get the approximate post-change pdf of $c[t]$ as
\begin{equation}
\label{eq:C_post}
\begin{split}
\tilde{f}_1^C =  K(1-e^{-\lambda_{\mc{S}}c[t]})^{K-1}  \lambda_{\mc{S}}e^{-\lambda_{\mc{S}}c[t]} (1-e^{-\lambda_{0}c[t]})^{N-K}\\ +  (N-K) (1-e^{-\lambda_{0}c[t]})^{N-K-1} \lambda_{0}e^{-\lambda_{0}c[t]}  (1-e^{-\lambda_{\mc{S}}c[t]})^{K},
\end{split}
\end{equation}
where $\lambda_0=\frac{1}{\sigma_n^2+ K\alpha^2 \sigma_x^2}$ and $\lambda_{\mc{S}}=\frac{1}{\sigma_n^2+ K\alpha^2 \sigma_x^2+ (1-\alpha^2)\sigma_x^2}$.\\[1.2ex]
Correlator CUSUM also does not use the fact that the unknown sparse signal $\vec{x}[t]$ has the same support for all $t \geq \nu$.

Correlator and Energy CUSUM do not provide a direct mechanism to identity the support set $\mc{S}$. Hence, at the 
time (say $T$) when the change is declared by Energy (or Correlator) CUSUM, we run a sparse signal recovery algorithm such as orthogonal matching pursuit (OMP) \cite{Tropp:TIT:07} on the observation $\vec{y}[T]$ and identify the support.

\subsubsection{Partial Support Estimation (PSE) CUSUM}
\label{sec:PSE}
A technique to detect sparse signals using a partial estimate of the support is presented in \cite{Wimalajeewa:TSIPN:17}. We combine this detection technique with CUSUM algorithm and employ the same for our change detection problem. Sparse signal recovery algorithms, like OMP, can be employed to obtain a partial estimate $\hat{\mc{S}}_p$  of support 
having cardinality $|\hat{\mc{S}}_p|=K_p$, where $1 \leq K_p \leq K$.  The sensing matrix in \cite{Wimalajeewa:TSIPN:17} is chosen to satisfy the condition 
$\vec{A}\vec{A}^*=\vec{I}_M$, i.e., its rows are orthonormal. The decision statistic considered here is the total power of the received signal $\vec{y}[t]$ projected on to the subspace spanned by the partial support estimate $\hat{\mc{S}}_p$ and normalized by the noise variance $\sigma_n^2$.
Specifically, the decision statistic is $p[t]=\frac{\norm{\vec{P}_{\vec{\hat{\mc{S}}}_p}\vec{y}[t]}_2^2}{\sigma_n^2}$, where the projection matrix is given by $\vec{P}_{\hat{\mc{S}}_p}= \vec{A}_{\hat{\mc{S}}_p} \big(\vec{A}_{\hat{\mc{S}}_p}^T \vec{A}_{\hat{\mc{S}}_p} \big)^{-1} \vec{A}_{\hat{\mc{S}}_p}^T$.
The pre-change distribution of $p[t]$ is $\chi^2_{K_p}$ which can be approximated as
$\tilde{f}^P_0=\mathrm{N}(K_p,2K_p)$, using CLT.
From \cite{Wimalajeewa:TSIPN:17}, the post-change distribution of $p[t]$ is $\chi^2_{K_p}(\mu_{\hat{\mc{S}}_p})$ with approximate non centrality parameter as $\tilde{\mu}_{\hat{\mc{S}}_p}=\frac{\E{\norm{\vec{P}_{\hat{\mc{S}}_p}\vec{A} \vec{x}[t]}}_2^2}{\sigma_n^2}=\frac{M K_p}{NK}\Big(1+ \frac{K- K_p}{M} \Big)\frac{\E\norm{\vec{x}[t]}_2^2}{\sigma_n^2}$. Here, $\E\norm{\vec{x}[t]}_2^2$ can be replaced by $K\sigma_x^2 $ when $\vec{C_x}=\sigma_x^2\vec{I}$. Thus, using a Gaussian approximation due to CLT, the post-change pdf is
$\tilde{f}^P_1 = \mathrm{N}\big(K_p+ \tilde{\mu}_{\hat{\mc{S}}_p}, 2(K_p + 2\tilde{\mu}_{\hat{\mc{S}}_p} )\big)$. The performance of PSE-CUSUM  relies heavily on the accuracy of the estimated partial support. 

\subsection{Sparsity Level Known, Both Support and Signal Variance Unknown}
\label{sec:unknown_sig_var}
In this section, we consider the case when both $\mc{S}$ and $\vec{C_x}$ are unknown and the signal covariance matrix can take the form $\vec{C_x}= \diag([\sigma_1^2,\cdots,\sigma_K^2])$. 
We assume that the sparsity level $K$ is known. Also, we assume the knowledge of the upper bound $\sigma_{\max}^2$ and the lower bound $\sigma_{\min}^2$ on the signal variances, such that $\sigma_{\min}^2 \leq \sigma_i^2 \leq \sigma_{\max}^2 \hspace{1mm} \forall i \in \mc{S}$. 
\subsubsection{Based on pdf-approximation}
We use this approach for Aggregate, Energy and Correlator CUSUM algorithms when support and signal variance is unknown. The only difference from the previous case is that the approximations for the 
post change pdfs of the decision statistics are obtained in terms of $\sigma_{\min}^2$ and $\sigma_{\max}^2$. 
We obtain these approximations based on the post-change pdf which gives the lowest KL distance from the pre-change pdf, in order to account for the worst case detection delay.

Using the derivations in \appref{G_sigvar}, the post-change distribution for Aggregate CUSUM in \eqref{G_post} is replaced with 
the approximation, 
\begin{eqnarray}
\label{eq:Gi_approx}
\tilde{f}_1^{G_{i}} =   \begin{cases}
    \mathrm{CN}(0,\sigma_n^2+ K\alpha^2 \sigma_{\min}^2), & i \notin \mc{S}, \\
    \mathrm{CN}(0,\sigma_n^2 +K\alpha^2 \sigma_{\min}^2+ (1-\alpha^2)\sigma_{\min}^2), & i \in \mc{S}.
  \end{cases}
\end{eqnarray} 

Using the derivations in \appref{E_sigvar}, for the Energy CUSUM metric $e[t]= \norm{\vec{y}[t]}^2_2$, 
the post-change distribution is $\tilde{f}_1^{E} = \mathrm{N} \big(\tilde{\mu}_{E},\tilde{\sigma}^{2}_{E} \big)$ and mean and variance of $e[t]$ are approximated as
\begin{eqnarray}
\label{eq:E_mean_approx}
\tilde{\mu}_E &=& K\phi_{\min}+ M\sigma_n^2, \\
\label{eq:E_var_approx}
\tilde{\sigma}^{2}_E &=& K\phi_{\min}^2 +2 \sigma_n^2  K\phi_{\min} +M(\sigma_n^2)^2,
\end{eqnarray}
where $\phi_{\min}=\max \big\{ 0,\sigma_{\min}^2(1-\alpha (K-1))\big\}$.

From the derivations in  \appref{G_sigvar}, for the Correlator CUSUM metric $c[t] =\norm{\vec{g}[t]}^2_\infty$, the post-change pdf $\tilde{f}_1^{C}$ is same as that in \eqref{C_post}, but with the parameters $\lambda_0$ and $\lambda_{\mc{S}}$ replaced with their approximations, 
\begin{eqnarray}
\label{eq:C0_approx}
\tilde{\lambda}_0 &=& \frac{1}{\sigma_n^2+ K\alpha^2 \sigma_{\min}^2},\\
\label{eq:Cs_approx}
\tilde{\lambda}_{\mc{S}} &=&\frac{1}{\sigma_n^2 + K\alpha^2 \sigma_{\min}^2+(1-\alpha^2) \sigma_{\min}^2}. 
\end{eqnarray}

\subsubsection{Based on parameter-estimation}
So far, we have followed the approach of approximating the post-change pdfs using bounds on the signal variances. 
In an alternative approach, we can adaptively estimate the unknown parameters \cite{Li:ICASSP:09,Singamasetty:ICASSP:17} in the post-change pdfs and compute the LLRs using these estimated parameters. One such approach, which we refer to as stochastic gradient decent (SGD) CUSUM, is described below.

Let $d[t]$ denote the decision statistic used for change detection and $d[t]=g_i[t], e[t]$ or $c[t]$ for Aggregate, Energy or Correlator CUSUM respectively.  Let $\theta$ be the unknown parameter in the post-change pdf. Let the actual value of $\theta$ be equal to $\bar{\theta}$ and 
let $\hat{\theta}$ be its estimate. Hence, 
$f_1^{D,\bar{\theta}}$ is the true post-change pdf. We define the LLR parameterized by $\theta$ as 
$L^{D,\theta}(d[t])=\log\frac{f_1^{D,\theta}(d[t])}{f_0^D(d[t])}$ and the corresponding CUSUM metric as $W^{D,\theta}[t]$. 
Let the regression function denoting the expected value of LLR at time instant $t$ be $V_t({\theta})=\E_{\nu}\{L^{D,{\theta}}(d[t])\}$.  Since expectation is taken w.r.t the true pdf of $d[t]$, it can be shown from \cite{Singamasetty:ICASSP:17} that, $V_t({\theta}) < 0$ for $t < \nu$ and 
$V_t({\theta})=\mathrm{D}_{KL}(f_1^{D,\bar{\theta}} \parallel f_0^{D} )-\mathrm{D}_{KL}(f_1^{D,\bar{\theta}} 
\parallel f_1^{D,{\theta}})$, for $t \geq \nu$. Here, $\mathrm{D}_{KL}(f_p\parallel f_q)$ is used to denote the KL distance from $f_q$ to $f_p$. The post-change ($t \geq \nu$) regression function $V_t(\theta)$ is maximized when ${\theta}=\bar{\theta}$, i.e., when the argument of the regression function is equal to the true value  of the parameter. 
This motivates a gradient descent based approach to estimate the unknown parameter, which is described below. 
 
At time $t$, the gradient of the regression function at the present estimate $\hat{\theta}[t]$ 
is $\frac{V_t(\hat{\theta}[t]+c)-V_t(\hat{\theta}[t]-c)}{2c}$, as the limit $c \rightarrow 0$. 
For SGD, using stochastic approximation principle \cite{kiefer,Singamasetty:ICASSP:17}, we replace the expectation (ensemble average) in $V_t(\theta)$ with an instantaneous approximation using the LLR from the actual values of $d[t]$ and $\hat{\theta}[t]$. Specifically, at time $t=0$, the estimate $\hat{\theta}[0]$ is initialized to zero.  With small positive constants $a$ and $c$, for $t \geq 0$, the estimate is updated as
\begin{eqnarray}
\label{eq:sgd_llr}
 \hat{\theta}[t+1]=\hat{\theta}[t]+a\frac{L^{D,\hat{\theta}[t]+c}(d[t])-L^{D,\hat{\theta}[t]-c}(d[t])}{c}.
\end{eqnarray}
Initializing $W^{D,\hat{\theta}[t]}[-1]=0$, the CUSUM metric is then computed as 
\begin{eqnarray}
\label{eq:sgd_cusum}
W^{D,\hat{\theta}[t]}[t]=\Big(W^{D,\hat{\theta}[t-1]}[t-1]+L^{D,\hat{\theta}[t]}(d[t]) \Big)^+
\end{eqnarray}
where $L^{D,\hat{\theta}[t]}(d[t])=\log\frac{f_1^{D,\hat{\theta}[t]}(d[t])}{f_0^D(d[t])}$ and the algorithm terminates according to rule $\mc{R}$ in \eqref{cusum}. 

We use the Aggregate, Energy and Correlator decision statistics for SGD-CUSUM and the implementation with each statistic is described below. 
From \eqref{G_i} in \appref{G_sigvar}, we consider $\theta=\alpha^2\sigma_{\textrm{sum}}^2+(1-\alpha^2)\sigma^2_i$ to be the unknown parameter for Aggregate-SGD-CUSUM. The estimate $\hat{\theta}[0]$ is initialized to zero. The post-change pdf of $g_i[t]$ for $i \in \mc{S}$, parameterized by $\theta$ is given by $ \tilde{f}_{1}^{G_i,\theta} = \mathrm{CN}(0, \sigma_n^2  + \theta)$.

For Energy-SGD-CUSUM, $\theta=\phi_{\min}$ in \eqref{E_mean_approx} and \eqref{E_var_approx}, is treated as the unknown parameter which must be initialized to zero at the start of the algorithm. The approximate parameterized post-change pdf is $\tilde{f}^{E,\theta}_1 =
\mathrm{CN}(\mu_E^\theta,(\sigma_E^\theta)^2)$ with $\mu_E^\theta  = K{\theta} + M\sigma_n^2$ and  
$(\sigma_E^\theta)^2 = K{\theta}^2 + 2\sigma_n^{2}K \theta +  M\sigma_n^4$.

For Correlator-SGD-CUSUM, the post-change pdf is given by \eqref{C} in \appref{G_sigvar}. We consider $\theta=\sigma^2_i$ to be the unknown parameter in ${\lambda}_{\mc{S}}$, so that
\begin{eqnarray*}
{\lambda}_{\mc{S}}^\theta &=&\frac{1}{\sigma_n^2 + K\alpha^2 \sigma_{\min}^2+ (1-\alpha^2) \theta}
\end{eqnarray*}
and initialize $\hat{\theta}[0]=0$. The parameter ${\lambda}_{0}$ in \eqref{C0_approx} does not depend on $\theta$.

The LLR and CUSUM metric in SGD CUSUM for all the above decision statistics is updated using \eqref{sgd_llr} and \eqref{sgd_cusum} respectively.

\subsection{Support, Signal Variance and Sparsity level are Unknown}
\label{sec:unknown_sparsity}
One must observe that the post-change distributions enlisted in \secref{Unknown_Support_Case} and \secref{unknown_sig_var} depend on the knowledge of the exact value of sparsity order $K$ of $\vec{x}[t]$. In this section, we address the case when support set $\mc{S}$, signal covariance $\vec{C_x}$ and sparsity level $K$ are unknown. However, we assume that an upper bound on the sparsity level $K_{\max}$ is known, such that $K \leq K_{\max}$. 
\subsubsection{Based on pdf-approximation}
We consider the decision statistics $d[t]$ equal to $g_i[t], e[t]$ and $c[t]$ for change detection using Aggregate, Energy and Correlator CUSUM respectively. With a decision statistic $d[t]$ and unknown sparsity $k$, we run the CUSUM algorithm \emph{parallelly} for all values of  $\lbrace k:1 \leq k\leq K_{\max}\rbrace$ and declare change based on the largest CUSUM metric. We use the approximate post-change pdfs of various decision statistics listed in \eqref{Gi_approx}, \eqref{E_mean_approx}, \eqref{E_var_approx},  \eqref{C0_approx} and \eqref{Cs_approx} for running CUSUM for each value of $k$. The LLR for a particular value of $k$ at time $t$ for a decision statistic $d[t]$ is computed as $L^{D,k}(d[t])=\frac{f_1^{D,k}(d[t])}{f_0^D(d[t])}$. The CUSUM metric  is updated as $W^{D,k}[t]=\big(W^{D,k}[t-1]+L^{D,k}(d[t])\big)^+$, with $W^{D,k}[-1]=0$. 

The parallel CUSUM change detection rule $\mc{R}$ is given by,
\begin{equation}
\label{eq:unknown_sparsity}
\begin{split}
\mc{R} = \left \{\begin{array}{ll} \text{Declare change} & \text{if}~ \underset{k  \in \{1, \cdots, K_{\max}\} }{\max} \hspace{1mm} W^{D,k}[t] > \tau, \\ 
\text{Continue} & \text{otherwise.} \end{array} \right. 
\end{split}
\end{equation}

\subsubsection{Based on parameter-estimation} 
With unknown sparsity level, the implementation of Aggregate-SGD-CUSUM remains same as that described in the previous section.

For Energy-SGD-CUSUM algorithm, we treat $\theta=K\phi_{\min}$ as the unknown parameter in \eqref{E_mean_approx} and \eqref{E_var_approx} which must be initialized to zero at the start of the algorithm. The approximate moments of the parameterized post-change pdf 
$\tilde{f}^{E,\theta}_1 =\mathrm{CN}(\mu_E^\theta,(\sigma_E^\theta)^2)$, are given by 
$\mu_E^\theta  = {\theta} + M\sigma_n^2$ and $(\sigma_E^\theta)^2 = \sigma^2_{\min} {\theta} + 2\sigma_n^{2} {\theta} +  M\sigma_n^4$,
where the term $K\phi_{\min}^2$ in \eqref{E_var_approx} is approximated as $K\phi_{\min}^2 = \theta \phi_{\min}\approx \theta \sigma^2_{\min}$.

Since the post-change pdf of the correlator statistic in \eqref{C_post} depends implicitly on sparsity $K$ and cannot be isolated in the form of a separate parameter, Correlator-SGD-CUSUM cannot be implemented for this case.

\section{Sensing Matrix Design} \label{sec:sdes} 
In this section, we  present deterministic constructions of sensing matrices based on designs from quantum information theory
\cite{Klappenecker:ISIT:05}.  
In addition to low mutual coherence, one of these quantum theoretic constructions has an additional structure in the sensing matrix 
which allows exact computation of the post-change pdfs of some decision statistics.

\subsection{Unitary Matrix}
For a unitary sensing matrix, the approximations for the post change pdfs of decision statistics are obtained by setting 
mutual coherence $\alpha$ to be zero. However, these approximations are exact for some scenarios, as given below.  
\begin{lemma} \label{lem:unit} 
With unitary sensing matrix, when the signal covariance $\vec{C_x}$ and sparsity level $K$ are known but the support $\mc{S}$ is unknown, the exact post-change distributions of the decision statistics $g_i[t]$, $e[t]$ and $c[t]$ are obtained by substituting $\alpha=0$ in \eqref{G_post}, \eqref{E_post} and \eqref{C_post} respectively.
\end{lemma}
\begin{proof}
Follows  from \appref{E_sigvar} and \appref{G_sigvar}.
\end{proof}

\subsection{Symmetric Informationally Complete Positive Operator Valued Measure (SIC POVM)}
\label{sec:SIC_POVM}
In a $d$-dimensional Hilbert space, SIC POVM is described by a set of $d^2$ rank-1 projectors, $\mc{P}_d=\lbrace \Pi_i=\frac{1}{d}\vec{a}_i \vec{a}_i^* : 1 \leq i \leq d^2 \rbrace$, 
with the property, 
\begin{eqnarray*}
\label{eq:SICPOVM}
\text{tr}(\Pi_i \Pi_{\ell})=\frac{1}{d^2}|\langle \vec{a}_i, \vec{a}_{\ell} \rangle |^2=\frac{1+\delta_{i{\ell}}d}{d^2(1+d)},
\end{eqnarray*}
where $\delta_{i{\ell}} = 1$ if $i={\ell}$ and $zero$, otherwise. From these SIC-POVM projectors, we obtain $d^2$ \emph{equi-angular} 
vectors $\{\vec{a}_i\}$ of unit length such that $|  \langle \vec{a}_i, \vec{a}_{\ell} \rangle | =  \frac{1}{\sqrt{d+1}},~ i \neq {\ell}$. Setting $M=d$ and $N \leq d^2$, we construct a sensing matrix 
$\vec{A}_{M \times N}$ using (a subset of) the SIC POVM vectors as its columns. With this construction, the magnitude of the inner product between any two distinct columns of $\vec{A}$ will be equal to $\alpha$. 

\begin{lemma}
For an $M \times N$ sensing matrix constructed using SIC POVM of dimension $M$, when the signal covariance $\vec{C_x}$ and sparsity level $K$ are  known but the support $\mc{S}$ is unknown, the exact post-change pdf of the decision statistic  $g_i[t]$ in Aggregate CUSUM algorithm is obtained by substituting $\alpha=\frac{1}{\sqrt{M+1}}$ in  \eqref{G_post}.
\end{lemma}
\begin{proof}
Follows from derivations in \appref{G_sigvar}.
\end{proof}

Though it is conjectured that SIC POVMs exist for every dimension $d$, the actual constructions for SIC POVMs are available only for some specific values of $d$ \cite{Scott:JMP:10,Grassl:JMP:17}. One of the popular techniques to obtain the SIC POVM vectors
$\{\vec{a}_i\}$ is to apply Weyl-Heisenberg (WH) displacement group operators to a fiducial vector 
\cite{Scott:JMP:10,Grassl:JMP:17}.    
The  Weyl-Heisenberg displacement group operators in dimension $d$ are generated by the cyclic shift operation $\hat{X}$ and its Fourier-transformed version  $\hat{Z}$ on a fixed orthonormal basis, say, 
$\lbrace \vec{e}_0,\vec{e}_1, \dots ,\vec{e}_{d-1} \rbrace $,
\begin{equation*}
\hat{X}=\sum\limits_{\substack{i=0}}^{d-1} \vec{e}_{i+1}\vec{e}_{i}^T \quad \text{and} \quad  \hat{Z}=\sum\limits_{\substack{i=0}}^{d-1} \omega_d^i \vec{e}_i \vec{e}_i^T,
\end{equation*}
where $\omega_d=\exp(\frac{2 \pi j }{d})$ is a complex primitive $d$-th root of unity and addition is modulo $d$. Without loss of generality, we can take $\lbrace \vec{e}_0,\vec{e}_1, \dots ,\vec{e}_{d-1} \rbrace $ to be the standard basis. 
The elements of the WH group, $\hat{X}^a\hat{Z}^b$, can be identified with pairs of integers, $(a,b) \in \mathbb{Z}_d \times \mathbb{Z}_d $.
The displacement operator is defined as $\hat{D}_{(a,b)}=\tau^{ab}\hat{X}^a\hat{Z}^b$ where the phase factor $\tau = -\exp{(\frac{\pi j }{d})}$. Now, the vectors in SIC POVM are constructed as
\begin{equation*}
\vec{a}_{(a,b)}  = \hat{D}_{(a,b)}  \vec{a}_{(0,0)}  , \quad  \text{where} \quad (a,b) \in \mathbb{Z}_d^2.
\end{equation*}
where $\vec{a}_{(0,0)}$ is referred to as the fiducial vector, which is available for some specific dimensions 
\cite{Scott:JMP:10,Grassl:JMP:17}.

\subsection{Mutually Unbiased Bases (MUB)}
In a Hilbert space of dimension $d$, MUBs are a set of $d+1$ unitary matrices (orthonormal bases), i.e.,   $\mc{M}_d=\big \lbrace\vec{U}_k=[ \vec{a}_0^k  , \dots , \vec{a}_{d-1}^k ] : 0\leq k \leq d  \big \rbrace$, such that the following properties are satisfied,
\begin{eqnarray*}
\label{eq:MUB}
\begin{split}
|\langle \vec{a}_{l}^{k}, \vec{a}_{r}^{q} \rangle|^2 &= \frac{1}{d}, \quad 
\begin{cases}
\forall k\neq q \in \lbrace 0,1,\dots,d \rbrace, \\
 \forall l,r \in \lbrace 0,1,\dots,d-1 \rbrace, 
\end{cases}
        \end{split} 
\end{eqnarray*} 
where $\vec{a}_{l}^{k}$ is the $l$-th column (basis vector) of $k$-th unitary matrix $\vec{U}_k$  and $\vec{a}_{r}^{q}$ is the $r$-th column (basis vector) of $q$-th unitary matrix $\vec{U}_q$ in $\mathcal{M}_d$.

We can design a sensing matrix $\vec{A}_{M \times N}$ with mutual coherence $\alpha=\frac{1}{\sqrt{d}}$ by using MUBs generated for dimension $d=M$. Let $\floor{\frac{N}{M+1}} = r$. We take $r$ columns from every unitary matrix $\lbrace \vec{U}_i \in \mc{M}_d : 0\leq i \leq d\rbrace $ and the 
remaining $N - r(M+1)$ columns are chosen, one each from the first $N - r(M+1)$ unitary (MUB) matrices. In this construction, the columns of the sensing matrix 
are uniformly distributed across all the MUB matrices so that, the magnitude of the inner product between any two randomly chosen columns 
is equal to $\alpha$ with a high probability. This construction of sensing matrix using MUB 
tries to mimic the \emph{equi-angular} effect of SIC POVM. Construction of MUBs is available for dimensions $d=p^n$ where $p$ is a prime number and $n$ is a non negative integer. The procedure to construct MUBs is detailed in \cite{Klappenecker:FQ7:04}. 

\subsection{Approximately Mutually Unbiased Bases (AMUB)}
To overcome the constraint on dimension for the construction of MUBs, we can use AMUBs to design sensing matrices as they can be generated for all dimensions. For any non prime dimension, $\mathbb{C}^d$, AMUB is a set of $d+1$ unitary matrices (orthonormal bases), i.e.,  $\mc{A}_d=\big \lbrace\vec{V}_k=[ \vec{a}_0^k  , \dots , \vec{a}_{d-1}^k ] : 0\leq k \leq d  \big \rbrace$, but at the cost of relaxing the condition,
\begin{align*}
\label{eq:AMUB}
    |\langle \vec{a}_{l}^{k}, \vec{a}_{r}^{q} \rangle|^2 = \left\{\begin{array}{lr}
       \frac{1+o(1)}{d}  \vspace{1mm} \hspace{2mm} \text{or}\\
 \frac{1+o(\log{d})}{d}
        \end{array}\right\}    \ \parbox{2in}{$\forall k\neq q \in \lbrace 0,1,\dots,d \rbrace$, \\
 $\forall l,r \in \lbrace 0,1,\dots,d-1 \rbrace,$}
\end{align*} 
where $\vec{a}_{l}^{k}$ is the $l$-th column (basis vector) of $k$-th unitary matrix $\vec{V}_k$  and $\vec{a}_{r}^{q}$ is the $r$-th column (basis vector) of  $q$-th unitary matrix $\vec{V}_q$ in $\mc{A}_d$. Sensing matrices can be constructed from AMUBs in the same manner as that 
from MUBs. Detailed procedure to construct AMUBs is presented in \cite{Shparlinski:TI:06}.

It is worth noting that the mutual coherence of the deterministic sensing matrix $\vec{A}_{M \times N}$ constructed from SIC POVM, MUB or AMUB is inversely proportional to $\sqrt{M}$.

\subsection{Random Sensing Matrices}
We can construct $M \times N$ sensing matrices by choosing the $M$ rows randomly from an $N$-dimensional discrete Fourier transform (DFT) matrix. Similarly we can also construct sensing matrices by generating random Bernoulli or complex Gaussian ensembles \cite{Candes:SPM:08}. It is important that the columns of these sensing matrices should be normalized to unit norm. Using randomly generated sensing matrices yields approximate distributions for all decision statistics.

\section{Simulation Results}
\label{Simulations}
In this section, we present the results obtained from Monte Carlo simulations of the algorithms described in \secref{cdet} and our inferences thereof, by comparing their performance based on worst case detection delay $D_w$ in \eqref{wd}  and average run length $T_r$ in \eqref{rl}. 
To find $D_w$, we fix the change point to be $\nu=20$ for all the simulations. 
We compare the performance with SNR defined as
\begin{equation}
\label{eq:SNR}
\text{SNR (dB)}=10\log_{10}{\frac{\E\norm{\vec{x}}^2_2}{\E\norm{\vec{n}}^2_2}} = 10\log_{10}{\frac{\sum_{i \in \mc{S}}\sigma^2_{i}}{M\sigma_n^2}}.
\end{equation}
In the simulations, we set $\sigma_n^2 = 1$.

\subsection{Effect of Sensing Matrix}
The dimensions of the sensing matrix $\vec{A}$ are fixed as $M=124,N=200$. We define the compression ratio of the sensing matrix as $c_r=\frac{M}{N}=0.62$. We set $K=5$ and the support indices of $\vec{x}[t]$ are randomly selected from $\lbrace 1,\cdots,N\rbrace$. The non-zero entries of $\vec{x}[t]$ are drawn randomly from $\mathrm{CN}(\vec{0},\vec{C_x})$, where $\vec{C_x}=\sigma_x^2\vec{I}$ and $\sigma_x^2$ is chosen according to \eqref{SNR} for a particular value of SNR. We consider deterministic sensing matrices constructed using SIC POVM, MUB and AMUB, using the procedure detailed in \secref{sdes}. For sensing matrix designed using MUB, we choose the prime power closest to 124 and take $M=125$. A truncated DFT matrix with randomly chosen rows is also 
considered. We also construct random sensing matrices with i.i.d. complex Gaussian (CN) entries and i.i.d. Bernoulli (BER) entries.

\begin{figure*}
  \centering
  \subfloat[Energy CUSUM]{\includegraphics[scale=0.562]{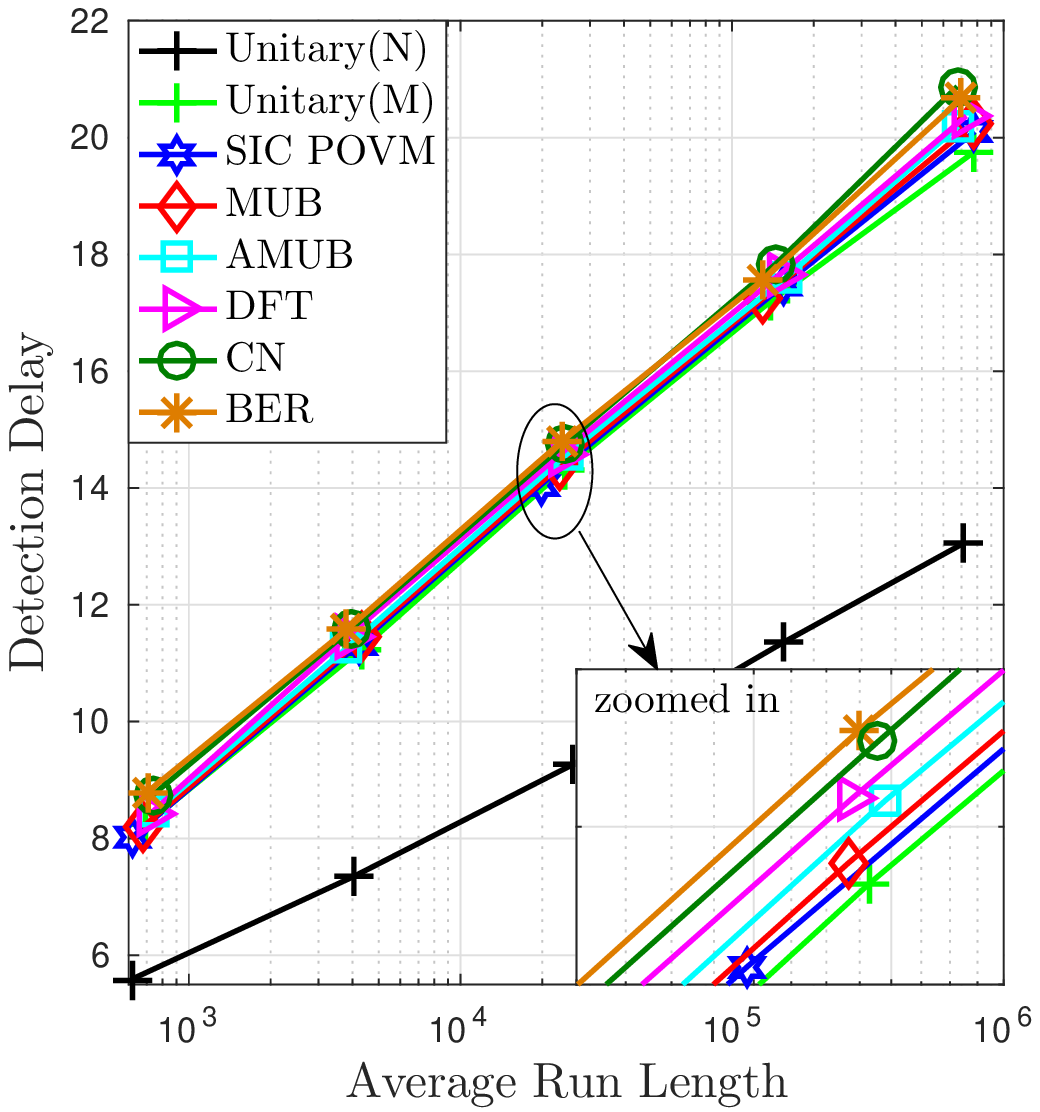}}
  \label{y_norm}
  \subfloat[Correlator CUSUM]{\includegraphics[scale=0.562]{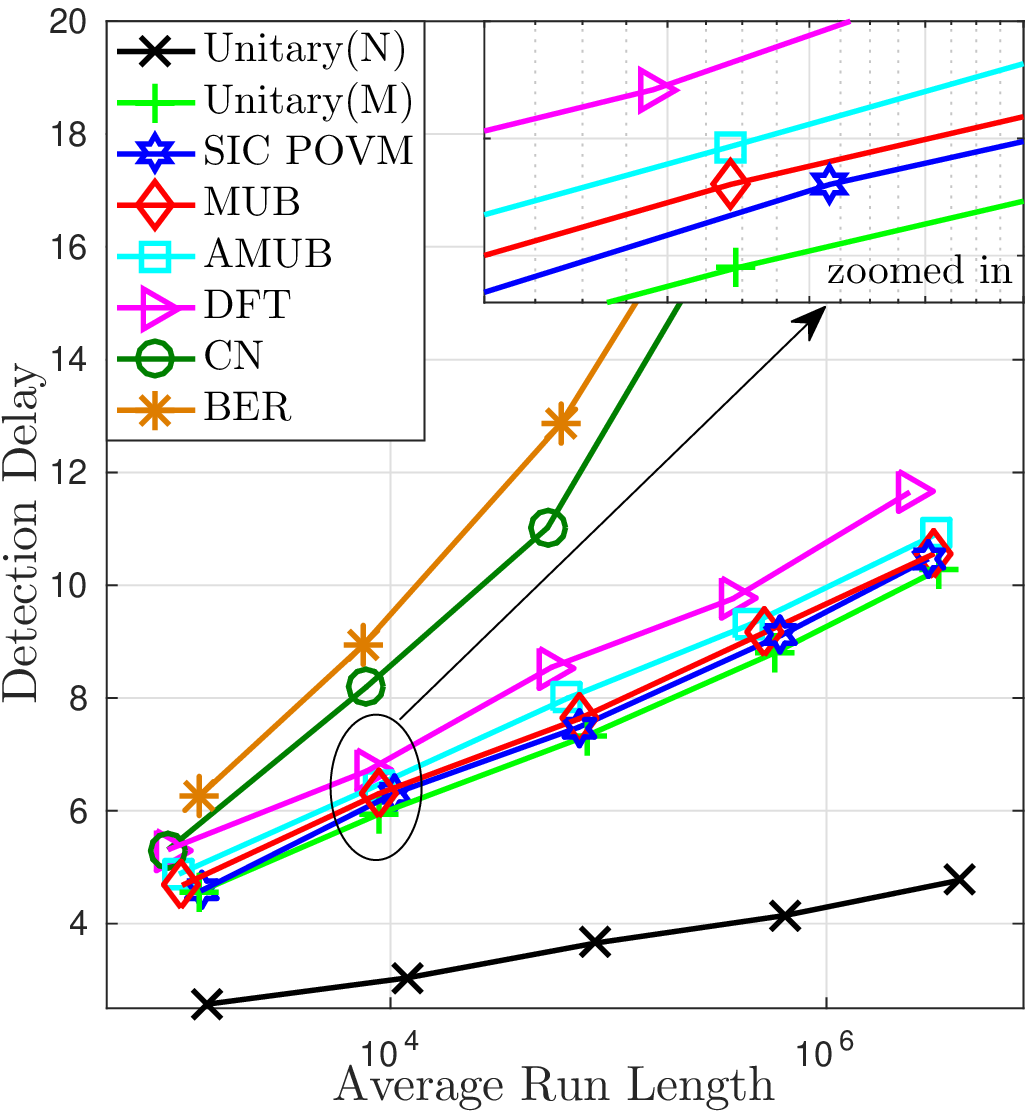}}
  \label{inf_norm}
   \subfloat[Aggregate CUSUM]{\includegraphics[scale=0.562]{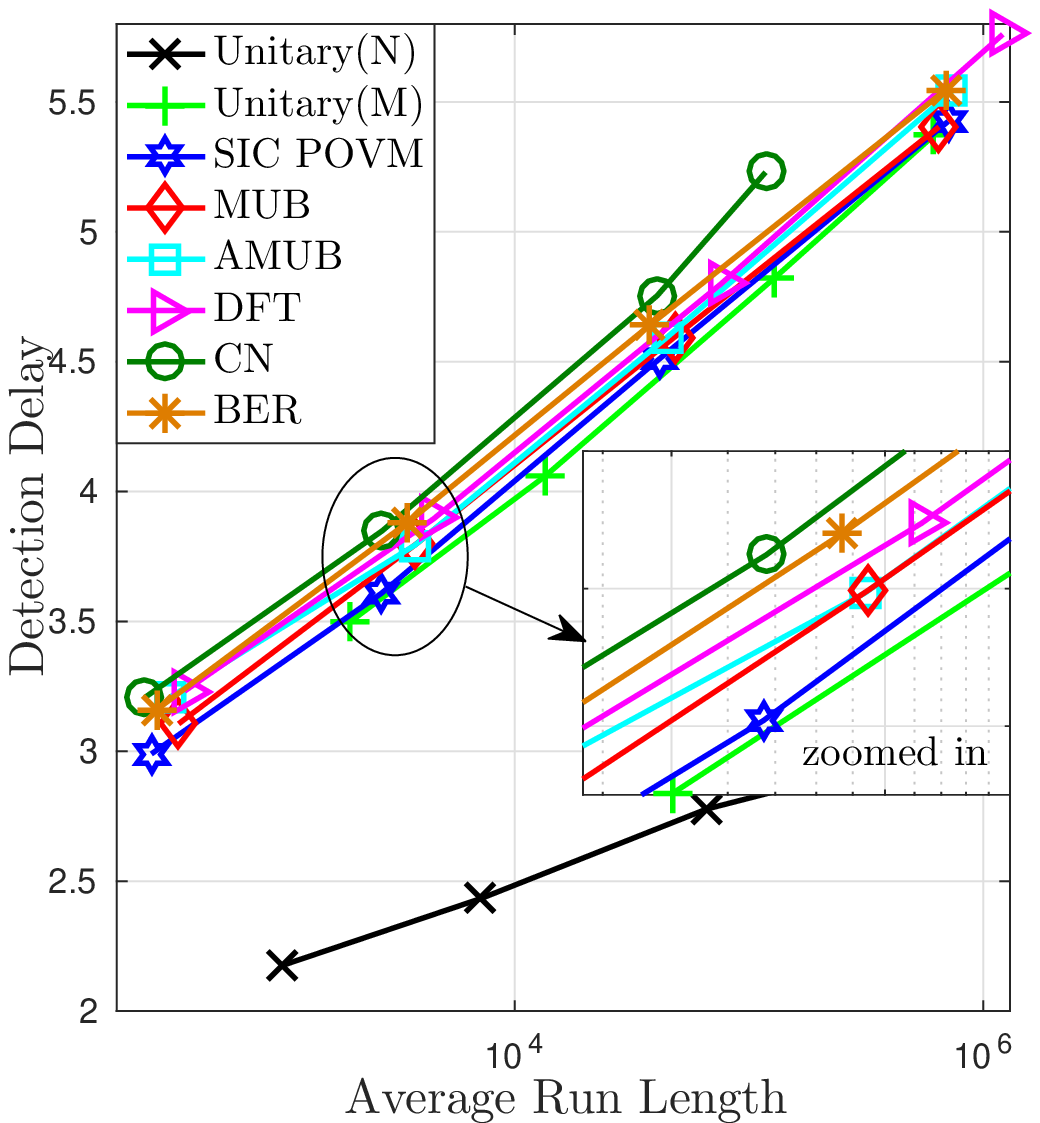}}
   \label{Aggregate}
   \caption{Comparison of various  $124 \times 200$ sensing matrix designs at SNR $= -10$ dB.}
   \label{fig:all_mat}
\end{figure*}

\begin{figure*}
  \centering
  \subfloat[Energy CUSUM]{\includegraphics[scale=0.562]{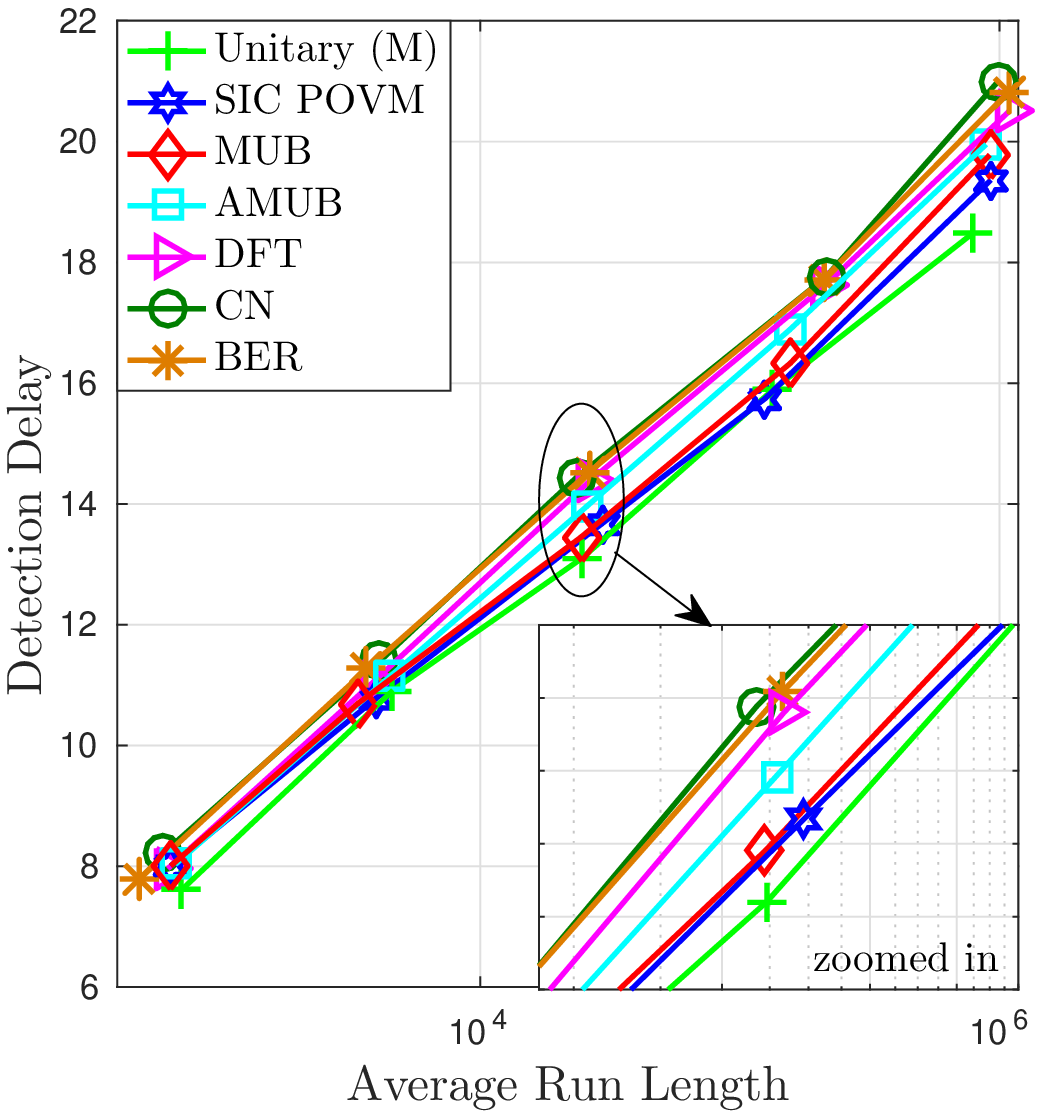}}
  \label{y_norm1}
  \subfloat[Correlator CUSUM]{\includegraphics[scale=0.562]{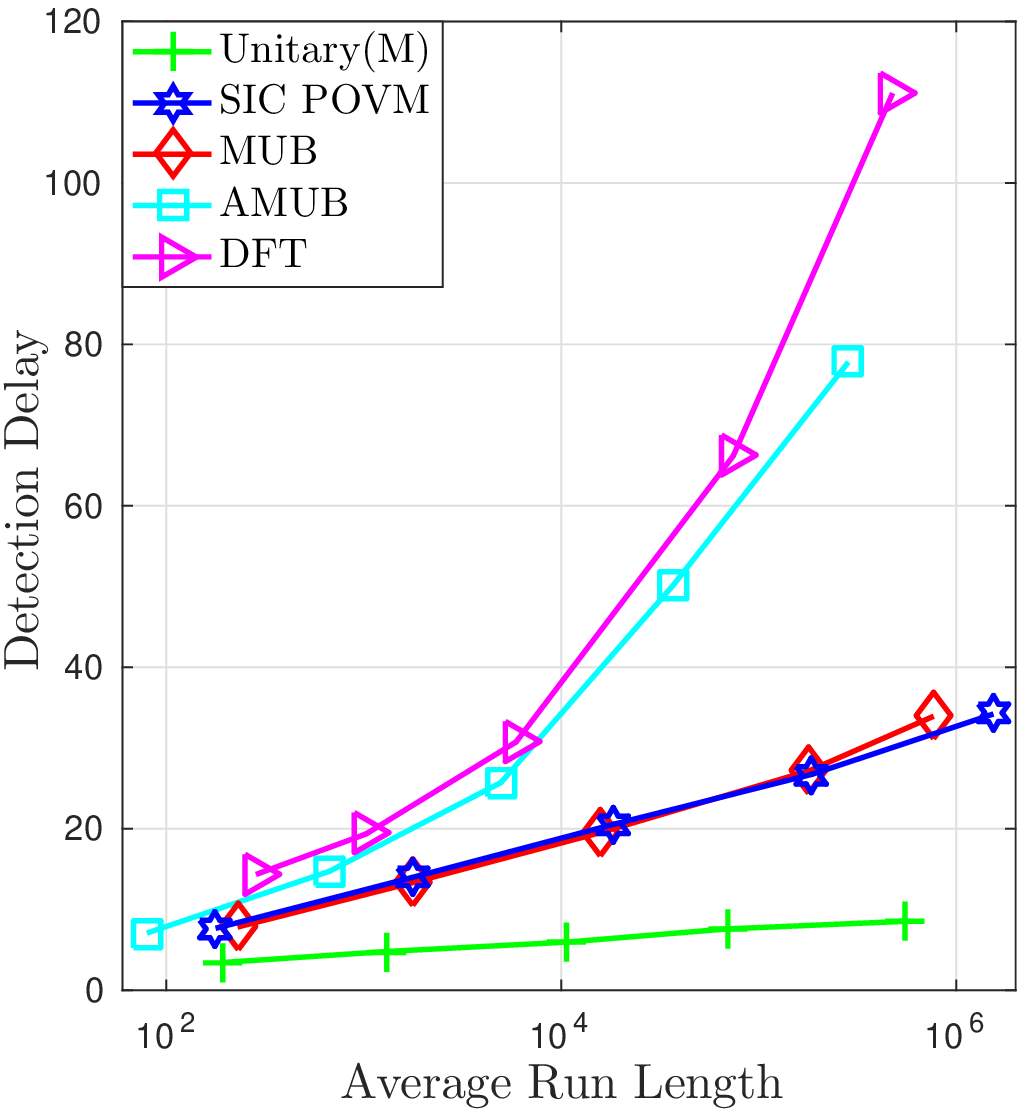}}
  \label{inf_norm1}
   \subfloat[Aggregate CUSUM]{\includegraphics[scale=0.562]{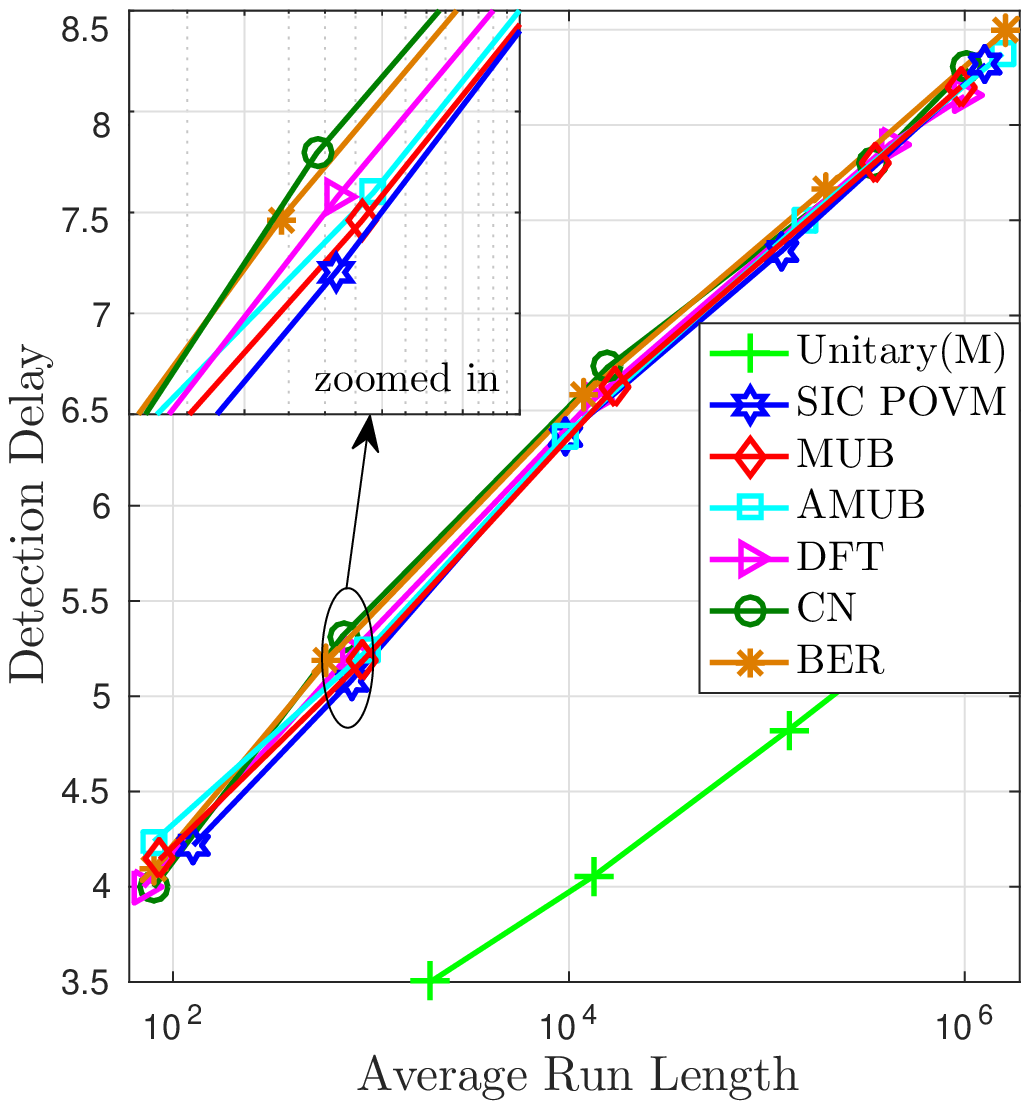}}
   \label{Aggregate1}
   \caption{Comparison of various  $124 \times 124^2$ sensing matrix designs at SNR $= -10$ dB.}
   \label{fig:all_mat_dd2}
\end{figure*} 

\Figref{all_mat} shows the performance of various sensing matrices in terms of $T_r$ versus $D_w$, at SNR$= -10$ dB. 
Unitary ($M$) denotes $M \times M$ unitary matrix that has the same number of measurements as $\vec{A}$ where as Unitary ($N$) denotes $N \times N$ unitary matrix that retains the same number of columns as $\vec{A}$. We infer that deterministic sensing matrices constructed from SIC POVM, MUB and AMUB give better performance as compared to random sensing matrices since they have a lower value of mutual coherence. In addition, these deterministic matrices yield near-exact post-change distributions for various decision statistics. Thus, as evident from  \Figref{all_mat}, the performances of SIC POVM and MUB sensing matrices, are closest to that of Unitary (M) matrix, followed by AMUB, DFT, complex Gaussian   and Bernoulli  sensing matrices. \Figref{all_mat_dd2} shows similar trends when the number of columns in $\vec{A}$ are increased to $N=124^2$ and the compression ratio  decreases to $c_r=0.008$. 
However, the detection delay $D_w$  for a specific $T_r$ increases for Aggregate and Correlator CUSUM. 
Also, Correlator CUSUM fails for random complex Gaussian and Bernoulli sensing matrices as the independence assumption on the 
entries of $\vec{g}[t]$ does not hold true due to the high mutual coherence of these random sensing matrices  when $c_r \ll 0.5$. 

\subsection{Comparison of Various Decision Statistics}
We fix the dimensions of $\vec{A}$ to be $M=124, N=200$ and compression ratio $c_r=0.62$. The columns of $\vec{A}$ are SIC POVMs generated for dimension, $M=124$. The sparsity level $K=5$ and signal covariance $\vec{C_x}=\sigma_x^2\vec{I}$, where $\sigma_x^2$ is chosen according to \eqref{SNR}.

\Figref{all_methods}  illustrates the $T_r$ versus $D_w$ plot for  various algorithms at SNRs $-20$ dB, $-10$ dB and $0$ dB. Aggregate CUSUM algorithm performs better than other support-oblivious algorithms at all SNRs, but with  additional computational complexity. 
We also observe that when SNR is $-10$ dB and above, Correlator CUSUM performs better than Energy CUSUM. 
On the other hand, when SNR is very low at $-20$ dB, the performance of Energy CUSUM becomes better than Correlator CUSUM, since 
entries in the correlation vector $\vec{g}[t]$ are highly corrupted by noise.  
The dashed lines in \Figref{all_methods} plot the performance of the parallel CUSUM rule in \eqref{unknown_sparsity} for the aforesaid algorithms when the signal covariance is known but the support and sparsity level are unknown. We fix the maximum sparsity level $K_{\max}=10$. A slight deterioration in performance of all decision statistics is observed as compared to the case when the value of $K$ is perfectly known. 
The performance of PSE CUSUM is poorer than other algorithms because accurate sparse support recovery $\vec{x}[t]$ becomes difficult to achieve with OMP algorithm at SNRs below $0$ dB. For PSE, the size of the recovered support set $K_p$ is kept identical to the true sparsity level $K$, which gives the best performance compared to any other value of $K_p$.
\begin{figure*}
  \centering
  \subfloat[SNR = 0 dB]{\includegraphics[scale=0.561]{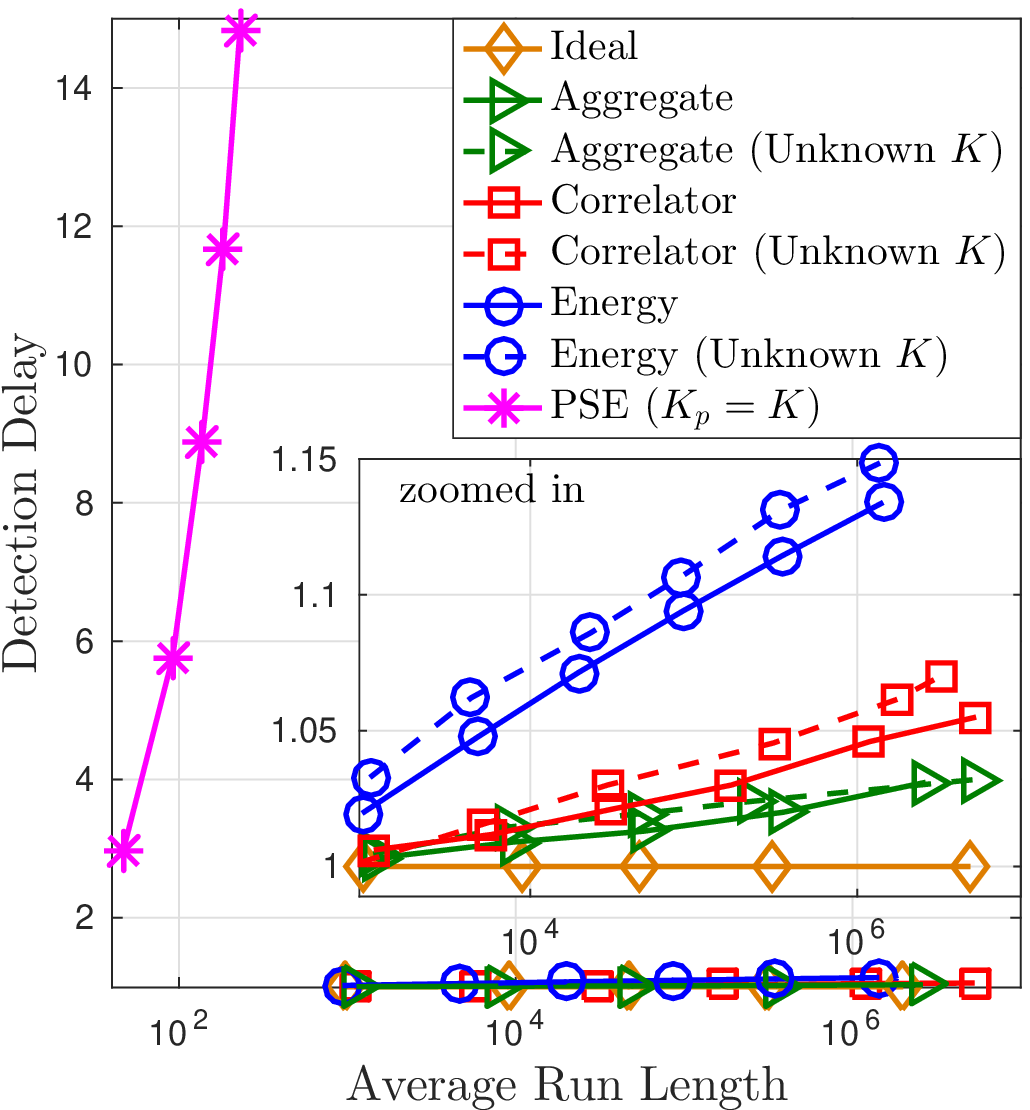}}
  \label{0db}
  \subfloat[SNR = -10 dB]{\includegraphics[scale=0.561]{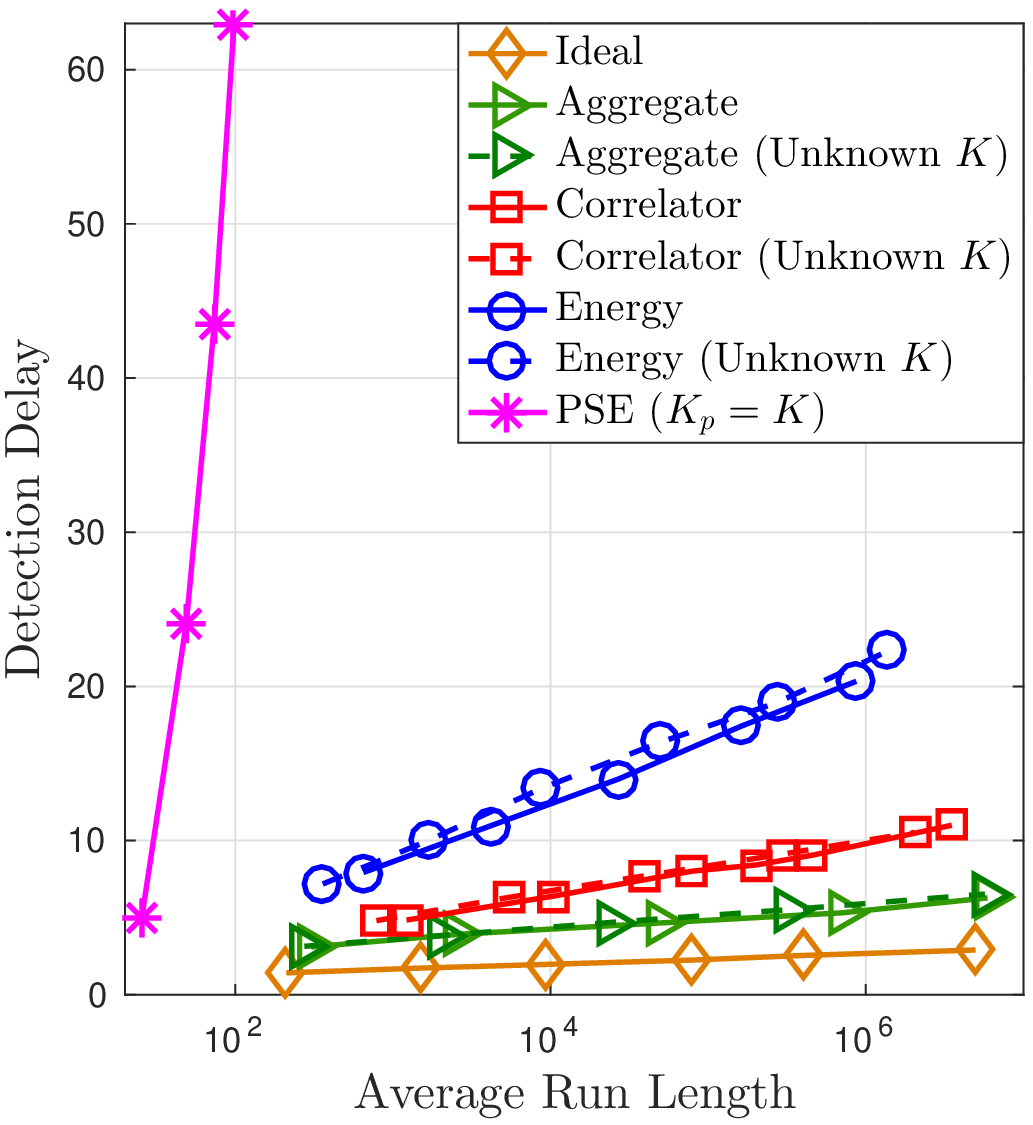}}
  \label{minus10}
   \subfloat[SNR = -20 dB]{\includegraphics[scale=0.564]{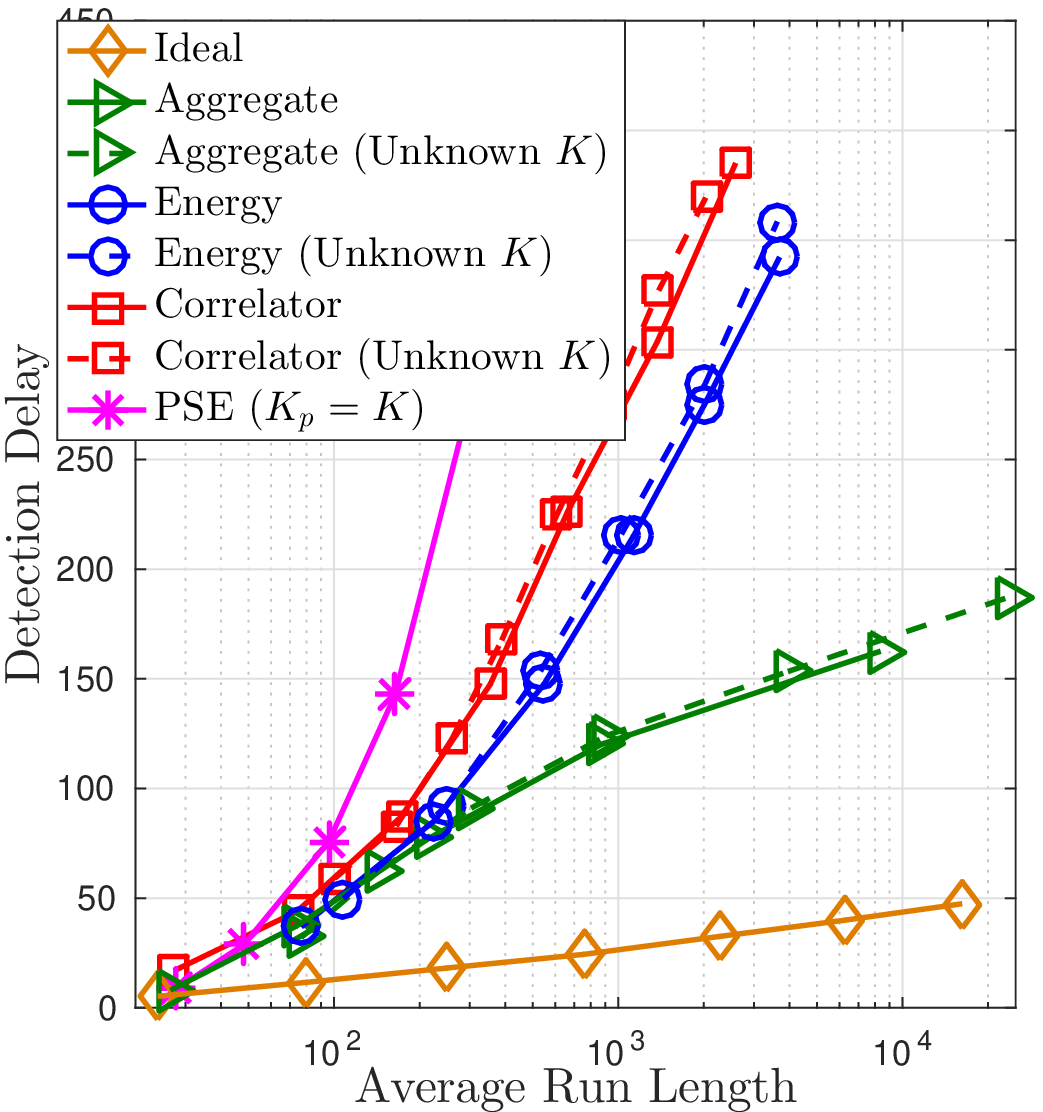}}
   \label{minus20}
   \caption{Comparison of various decision statistics using $124 \times 200$ SIC POVM sensing matrix at different SNRs.}
   \label{fig:all_methods}
\end{figure*}
\begin{figure*}
  \centering
  \subfloat[SNR = 0 dB]{\includegraphics[scale=0.57]{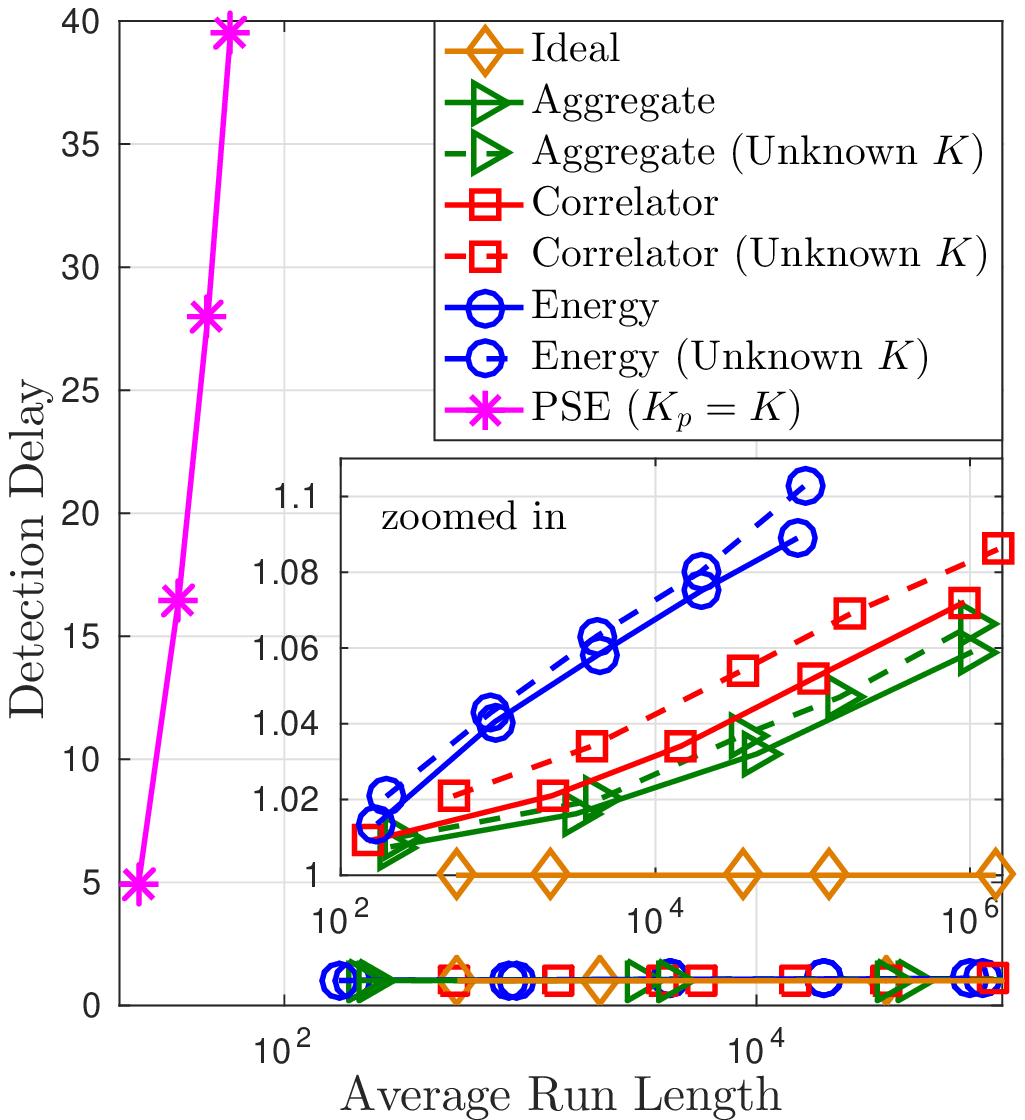}}
  \subfloat[SNR = -10 dB]{\includegraphics[scale=0.57]{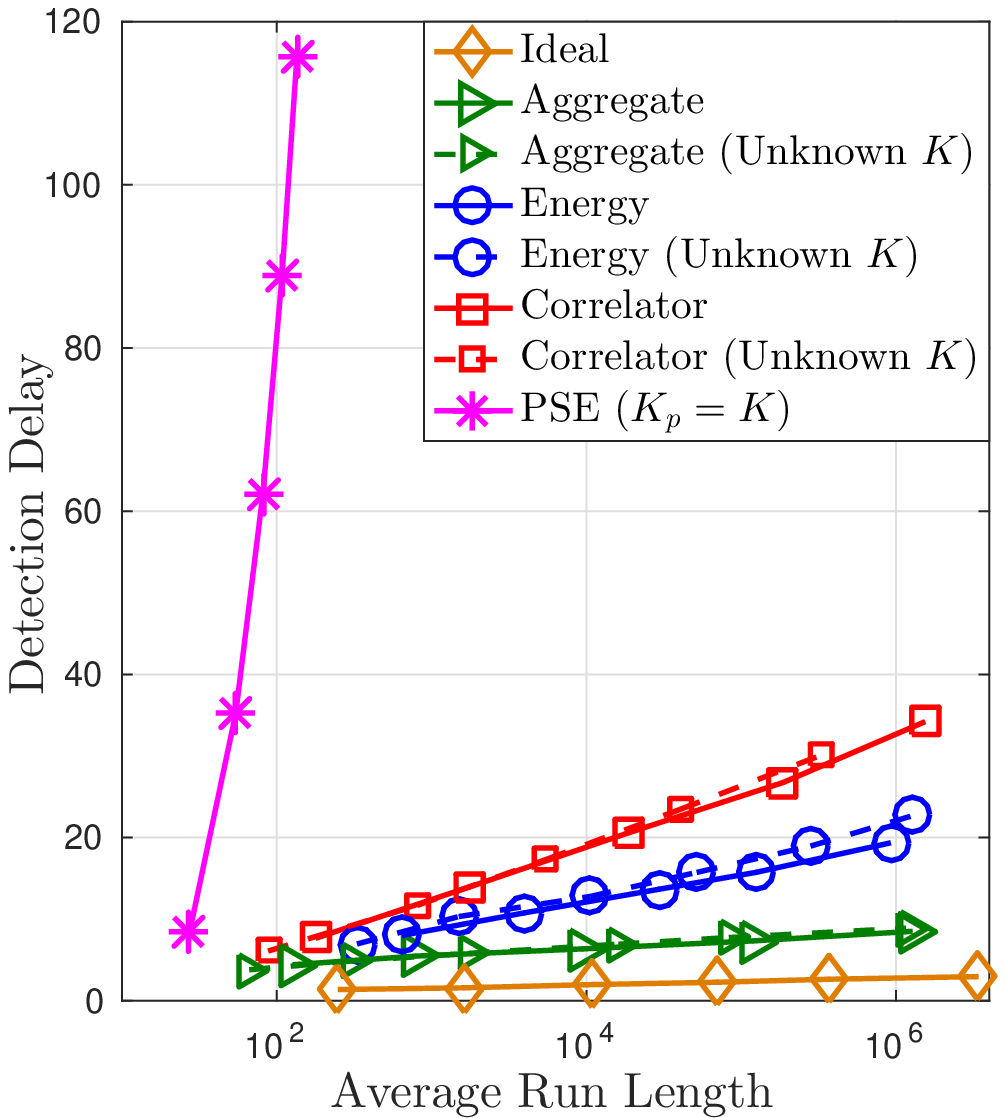}}
  \subfloat[SNR = -20 dB]{\includegraphics[scale=0.57]{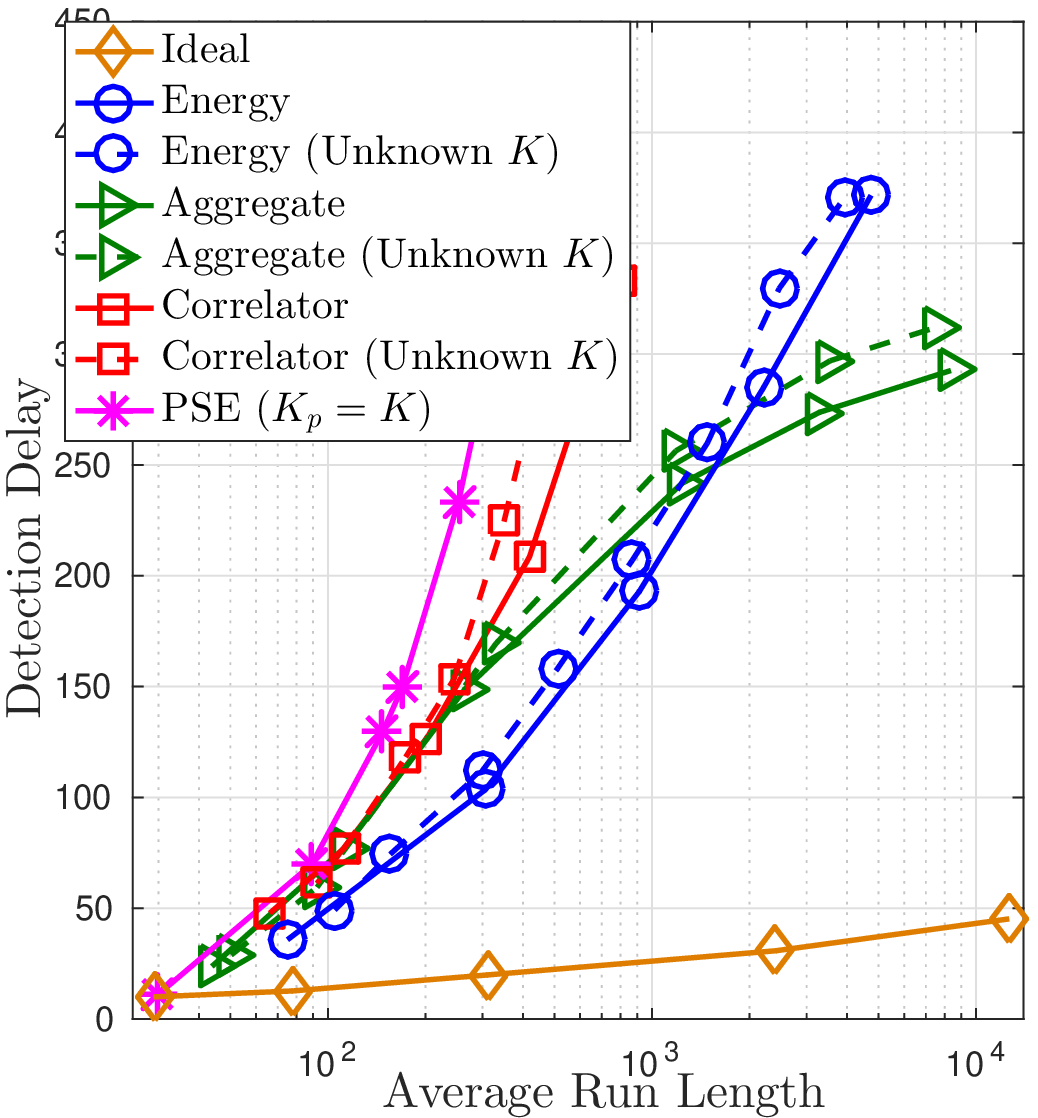}}
     \caption{Comparison of various decision statistics using $124 \times 124^2$ SIC POVM sensing matrix at different SNRs.}
      \label{fig:dd2_all}
\end{figure*}
We also consider the SIC POVM sensing matrix of size $M=124,~N=124^2$, with compression ratio $c_r=0.008$. All the other parameters are kept same as above. In \Figref{dd2_all}(a), at SNR $0$ dB, the detection performance of various algorithms follows the same trend as that shown in \Figref{all_methods}(a), but the overall detection delay $D_w$ is larger for a given $T_r$. 
In \Figref{dd2_all}(b), we see that Aggregate CUSUM algorithm performs close to Ideal CUSUM at SNR$=-10$ dB. At SNR$=-20$ dB, Aggregate CUSUM performs better than Energy CUSUM for higher values of the $T_r$ and poorer than Energy CUSUM for relatively smaller values of the $T_r$. 
Both Aggregate and Correlator CUSUM inherently assume/approximate that the entries in the correlation vector $\vec{g}[t]$ are independent. However, when $N$ is very large as compared to $M$, this independence approximation becomes inaccurate and the performance of these algorithms suffers. 

 \begin{figure}
  \centering
  \subfloat[$ M=124,N=200,K=5$]{\includegraphics[scale=0.46]{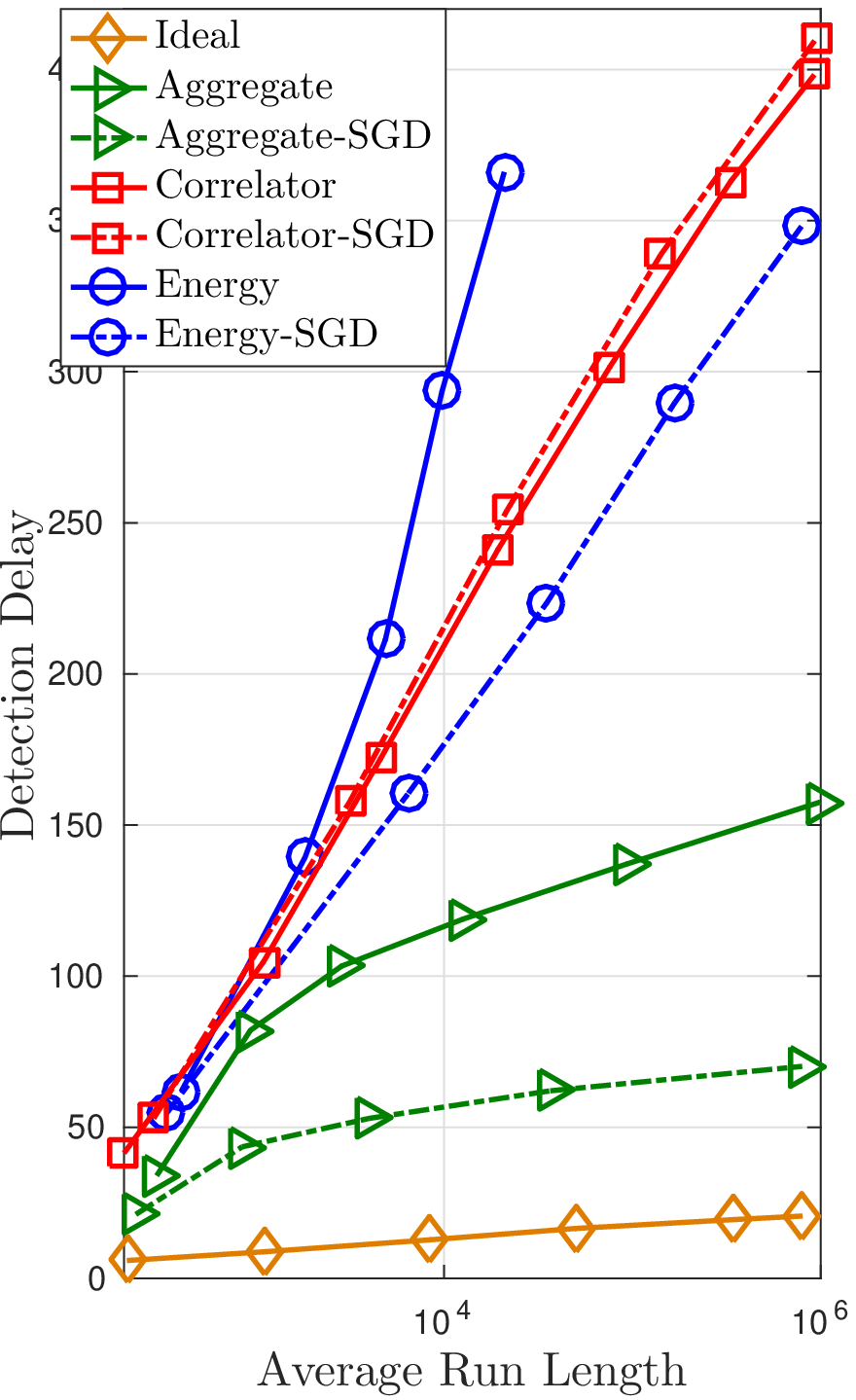}} 
  \subfloat[$M=124,N=124^2,K=5$]{\includegraphics[scale=0.46]{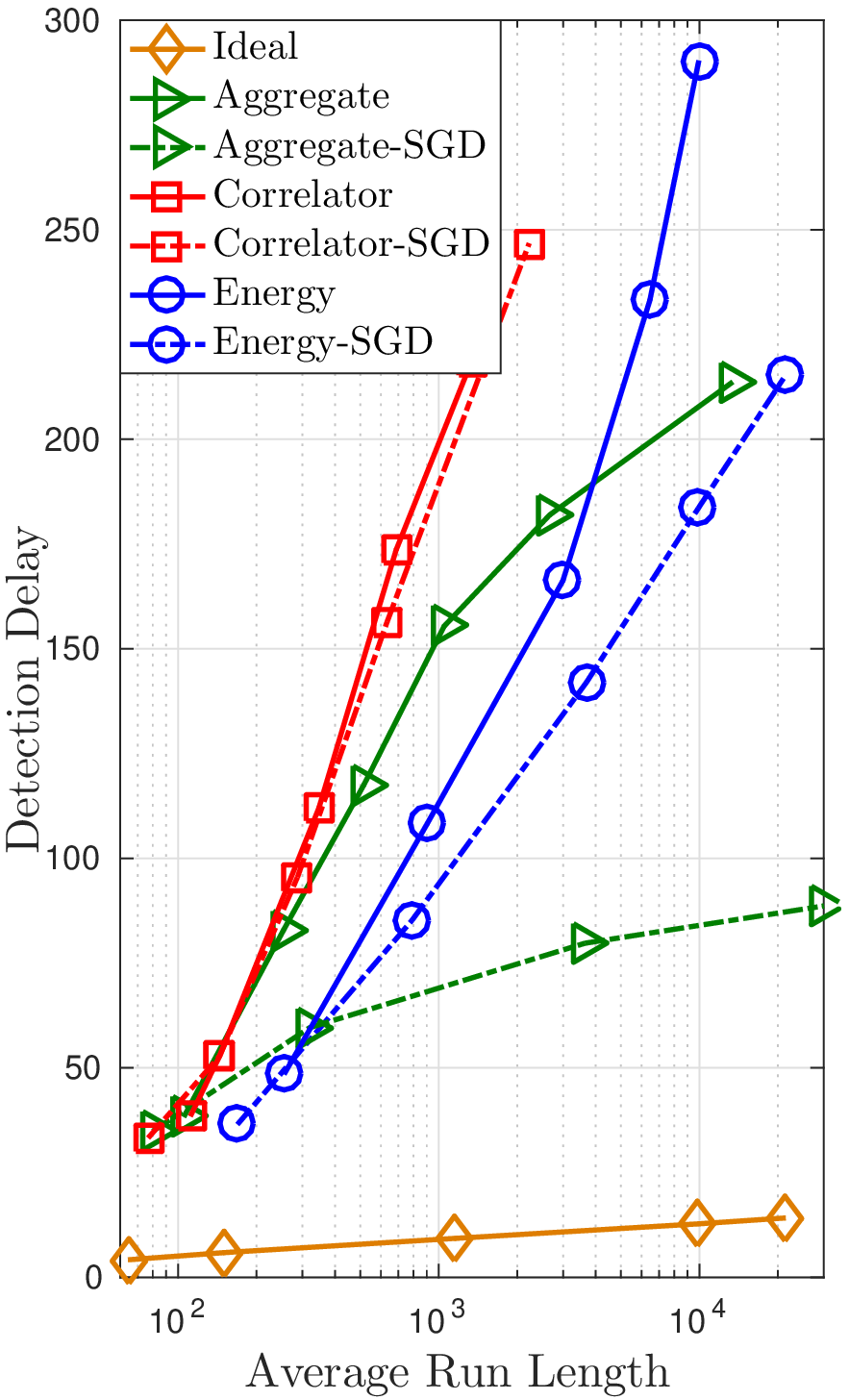}}
     \caption{Comparison of various decision statistics using SIC POVM sensing matrix with  $ \sigma^2_{\min}=0.1,\sigma^2_{\max}=1$ and $\sigma_n^2=1$ when $\mc{S}$ and $\vec{C_x}$ are unknown.}
      \label{fig:sigvar_all}
\end{figure}

\subsection{Unknown Support, Signal Variance and Sparsity level} 
Now, we address various cases regarding the knowledge of support, sparsity level and signal variance and present the corresponding simulation 
results. When the signal variance is unknown, we generate the non-zero entries of $\vec{x}[t]$ with (unequal) variances which are uniformly distributed in the interval $[\sigma^2_{\min},\sigma_{\max}^2]$. When the signal variances are unknown, SGD CUSUM algorithm (which tries
to estimate the unknown parameters) as discussed in \secref{unknown_sig_var} can be used in addition to the detection techniques that are based on approximating the
post-change pdf. 
For the SGD CUSUM, we set step size $a=0.01$ and window length $c=0.05$ in \eqref{sgd_llr}, in all the simulations.

When the sparsity level $K$ is assumed to be known a priori but the support $\mc{S}$ and signal covariance $\vec{C_x}$ are unknown, we employ the techniques given in \secref{unknown_sig_var} and plot their performance in \Figref{sigvar_all}. 

When the support $\mc{S}$, signal covariance $\vec{C_x}$ and the sparsity level $K$ of the signal are unknown, the methods outlined in \secref{unknown_sparsity} are used and the results are shown in \Figref{sigvar_sparsity_all}.

\begin{figure}
  \centering
  \subfloat[$ \sigma^2_{\min}=0.1,\sigma^2_{\max}=1$]{\includegraphics[scale=0.42]{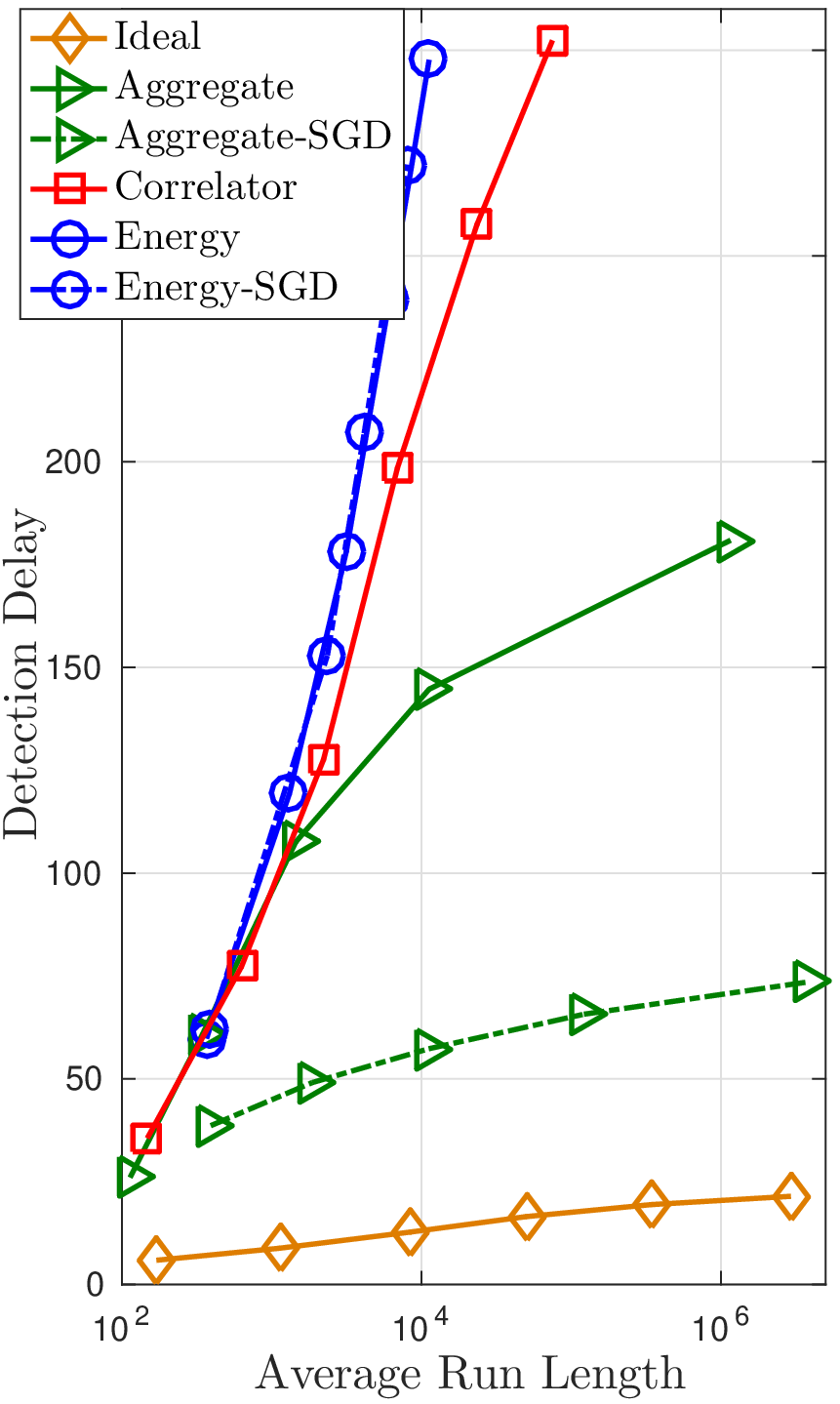}} 
  \subfloat[$ \sigma^2_{\min}=0.05,\sigma^2_{\max}=.5$]{\includegraphics[scale=0.42]{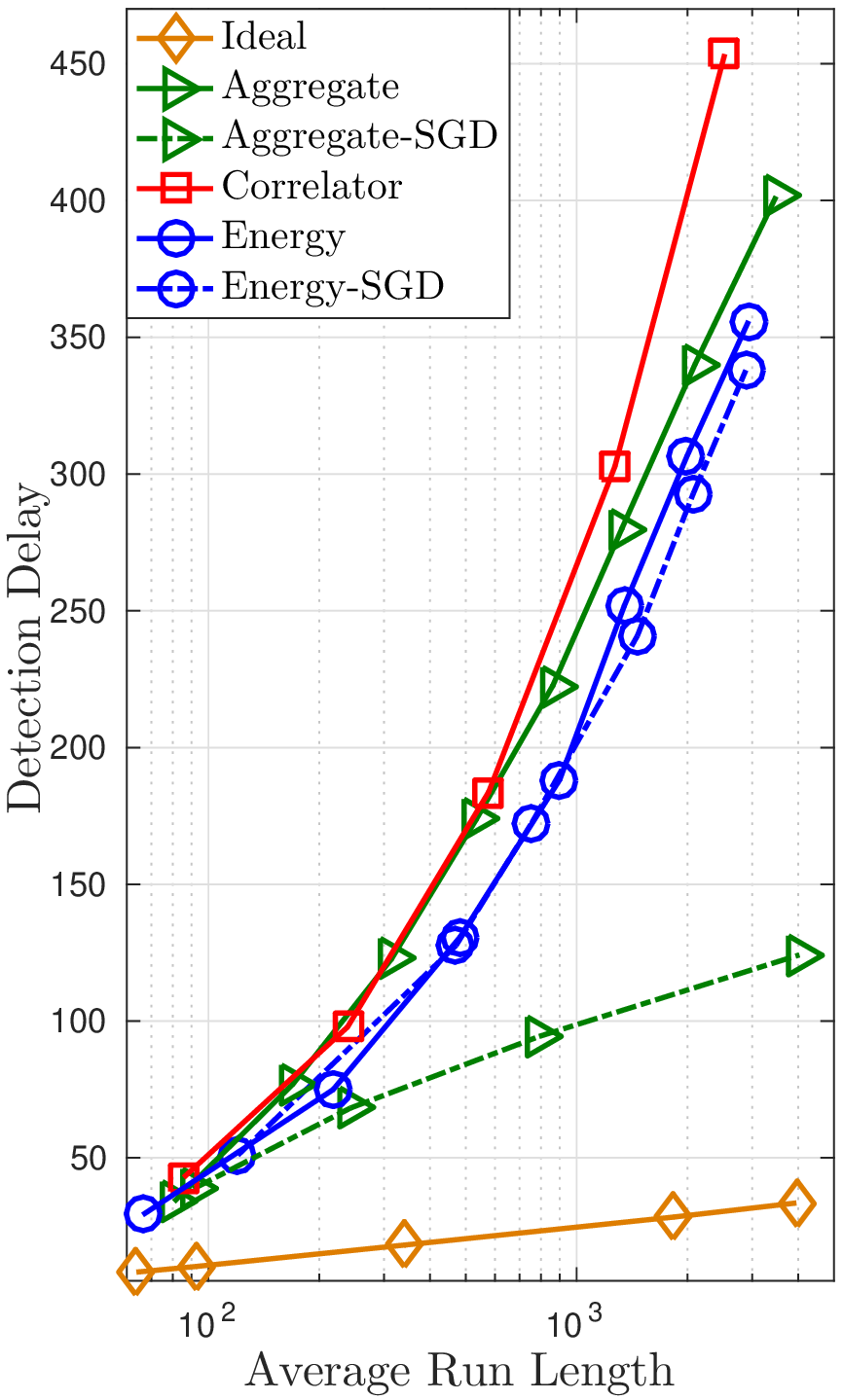}}
  \caption{Comparison of various decision statistics using SIC POVM sensing matrix with  $M=124, N=200, K=5$ and $\sigma_n^2=1$ when $\mc{S},\vec{C_x}$ and $K$ are unknown.}
  \label{fig:sigvar_sparsity_all}
\end{figure}

Some important observations are highlighted below.

\begin{enumerate}

\item If the signal variance is high $\sigma_{\max}^2 = \sigma_n^2$ or the compression ratio is large ($c_r > 0.5$), Aggregate CUSUM performs 
better than Energy and Correlator CUSUM.

\item If the signal variance is small $\sigma_{\max}^2 < \sigma_n^2$ or the compression ratio is small $c_r \ll 0.5$, Energy CUSUM 
based on the signal energy $\|\vec{y}[t]\|_2^2$ performs, in general, better than Correlator and Aggregate CUSUM, 
which use the correlation statistics $\vec{g}[t]$. 

\item In general, SGD CUSUM performs better than the corresponding pdf approximation based counterparts.

\item In most cases, Aggregate-SGD-CUSUM gives the best performance.

\end{enumerate}

\subsection{Percentage of Sparse Recovery}
In addition to detecting change, we are also interested in recovering the support of the signal $\vec{x}[t]$, when the change is detected. Note that, Aggregate CUSUM has an inherent mechanism to find the support by selecting the locations corresponding to the $K$-largest CUSUM metrics \eqref{pcusum}, when the change is detected. On the other hand, for Energy and Correlator CUSUM, once the change is detected, we run OMP algorithm to find the support. We define percentage of support recovery to be the fraction of the support that is recovered correctly. From \Figref{recovery}, for an average run length of $5 \times 10^3$, Aggregate CUSUM gives higher percentage of recovery than OMP at SNR $=-20$ dB and $-10$ dB while OMP is better at SNR $=0$ dB. As we increase the average run length, the change detection delay increases and this increases the percentage support recovery by Aggregate CUSUM. Support recovery using OMP, on the other hand, does not depend on the average run length. We also note that, the percentage recovery is better when the compression ratio $c_r$ is high.
\begin{figure}[htp]
    \centering
    {\includegraphics[scale=0.45]{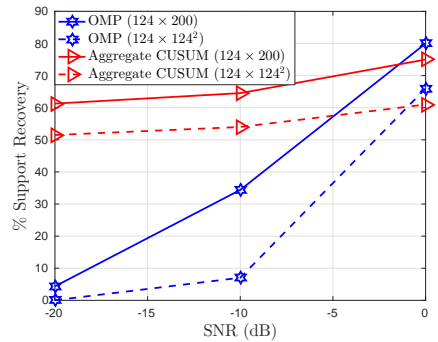}}
 \caption{Percentage recovery with constant  $T_r=5 \times 10^3$.}
    \label{fig:recovery}
\end{figure}
\subsection{Massive Random Access with Timing Offset}
Consider the user activity detection problem in massive random access application. Each user $i \in \{1,\cdots,P\}$ in the network is assigned a unique code/sequence, say, $\vec{a}_i$ of length $M$. In typical scenarios, there are relative timing offsets (due to propagation delays) in the reception of signals transmitted from different users. We assume that the timing offsets of all the users in the network are upper bounded by 
$\Delta \in \mathbb{Z}^{+}$. Suppose $i^{th}$ user has a delay of  $\delta_i \in \{0,\cdots,\Delta\}$, the code sequence 
received at the central node from the $i^{th}$ user will be 
$\vec{a}_{i,\delta_i}^T=[\vec{0}_{\delta_i}^T \hspace{2mm} \vec{a}_i^T \hspace{2mm} \vec{0}_{\Delta-\delta_i} ^T]$. 
The timing offset values of users are not usually available at the central node. Let $P$ be the number of users in the system. We form an augmented sensing matrix $\vec{A}_{\Delta}$ of size $(M+\Delta) \times P (\Delta+1)$
which contains all sequences (including all the possible timing offsets) of the form 
$\{\vec{a}_{i,\delta_i}\},~1\leq i \leq P, 0 \leq \delta_i \leq \Delta$. After the change point, a subset of users
$\mc{S}$ become active so that the observations are $\vec{y}[t] = \vec{A}_{\Delta} \vec{x}_{\Delta}[t] + \vec{n}[t]$, where the locations of the
non zero entries in $\vec{x}_{\Delta}[t]$ indicate the active users and their corresponding timing offsets. 
  
We consider code constructions based on SIC POVM and compare it with Gold codes. For SIC-POVM of length $M$, 
there are a total $M^2$ sequences. From the construction in \secref{SIC_POVM},
all the cyclic shifts of a SIC POVM sequence are also SIC POVM sequences. However, since the augmented matrix $\vec{A}_{\Delta}$ contains 
time shifts up to $\Delta$ for each sequence, we can use only $M\lfloor \frac{M}{\Delta+1} \rfloor$ sequences as valid codes (in order to 
avoid the scenario where code of one user is highly correlated with the time delayed code of another user). Hence, the maximum number of users
that can be accommodated with $M$ length SIC POVM codes is $P = M\lfloor \frac{M}{\Delta+1} \rfloor$. 
In a very similar manner, we can construct $\vec{A}_{\Delta}$ from the cyclic shifts of bipolar Gold codes.  
Gold codes exist for $M=2^n-1,~n\in \mathbb{Z}^{+}$ and we  denote them as $\{\vec{b}_i,~i=1,\cdots,M\}$. With $\vec{b}_{i,\tau}$ denoting
the cyclic shift of $\vec{b}_i$ by $\tau$, it has been shown in \cite{Gold:TIT:67} that 
$|\langle\vec{b}_{i,\tau_1},\vec{b}_{k,\tau_2}\rangle| \leq r(n)$ where
\begin{align*}
 &r(n)= 
\begin{cases}
\frac{2^{\frac{n+1}{2}}+1}{M},\ \text{if $n$ is odd},\\[0.5ex]
\frac{2^{\frac{n+2}{2}}+1}{M}, \ \text{if $n$ is even}.
\end{cases}
\end{align*}

In our simulations, we fix the number of users to be $P=1500$ and the maximum admissible timing offset as $\Delta=8$. 
We set $M=124$ for SIC POVM and $M=127$ for Gold codes. For this scenario, the mutual coherence of the augmented sensing matrix 
$\vec{A}_\Delta$ for SIC POVM and Gold code based constructions is $\alpha=0.1564$ and $\alpha=0.1969$, respectively. 
\Figref{mra}(a) shows that the user activity detection delay versus average run length is nearly the same for both constructions, 
with Energy and Aggregate CUSUM. In \Figref{mra}(b), we show the percentage of correctly identified users at the point when change 
is detected. Aggregate CUSUM performs better support recovery at low SNR while OMP performs better at high SNR. Also, the SIC POVM based sensing matrix performs better than those constructed using Gold codes.

\begin{figure}
    \centering 
\subfloat[Performance at SNR$=-10$ dB]{\includegraphics[scale=0.42]{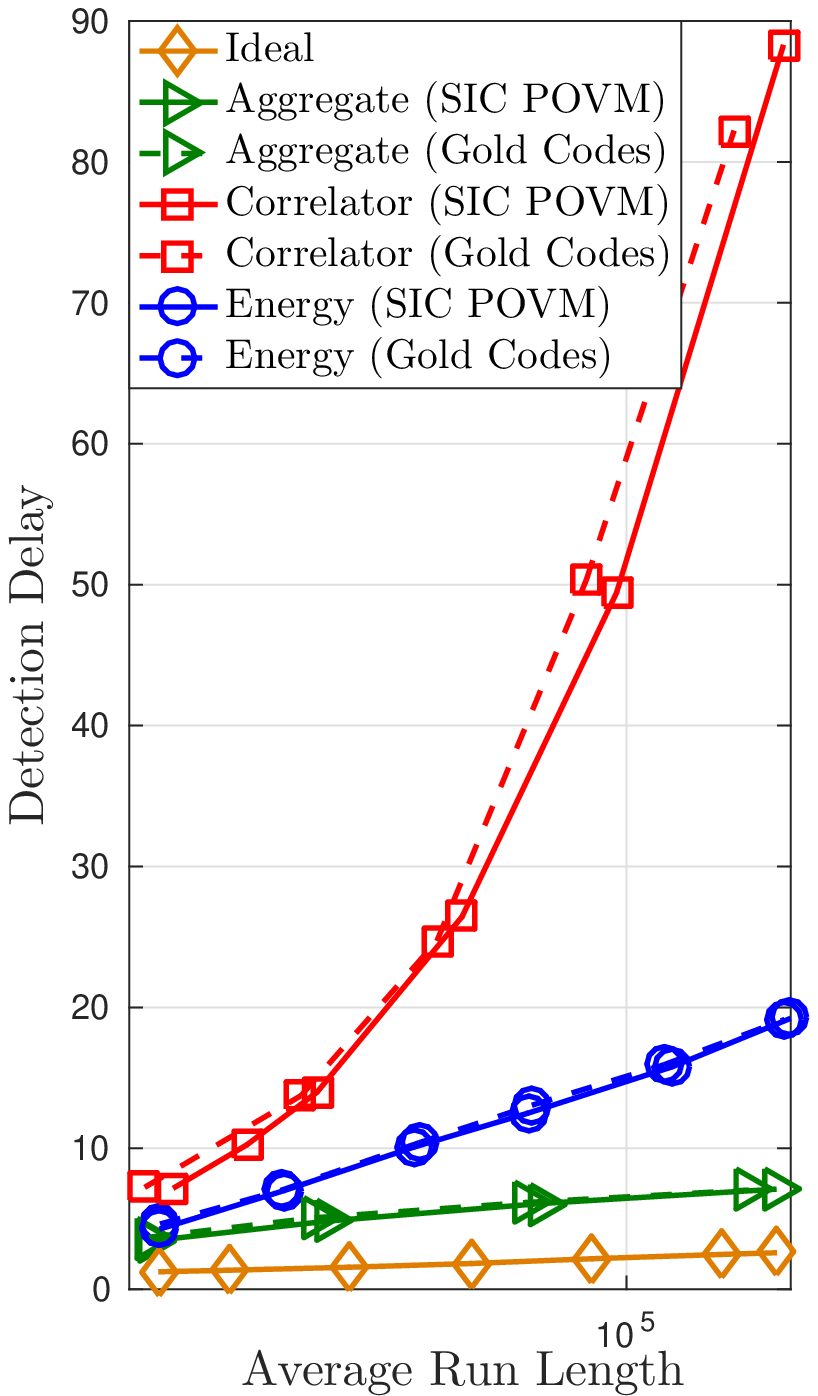}}
\subfloat[User identification v/s SNR ]{\includegraphics[scale=0.42]{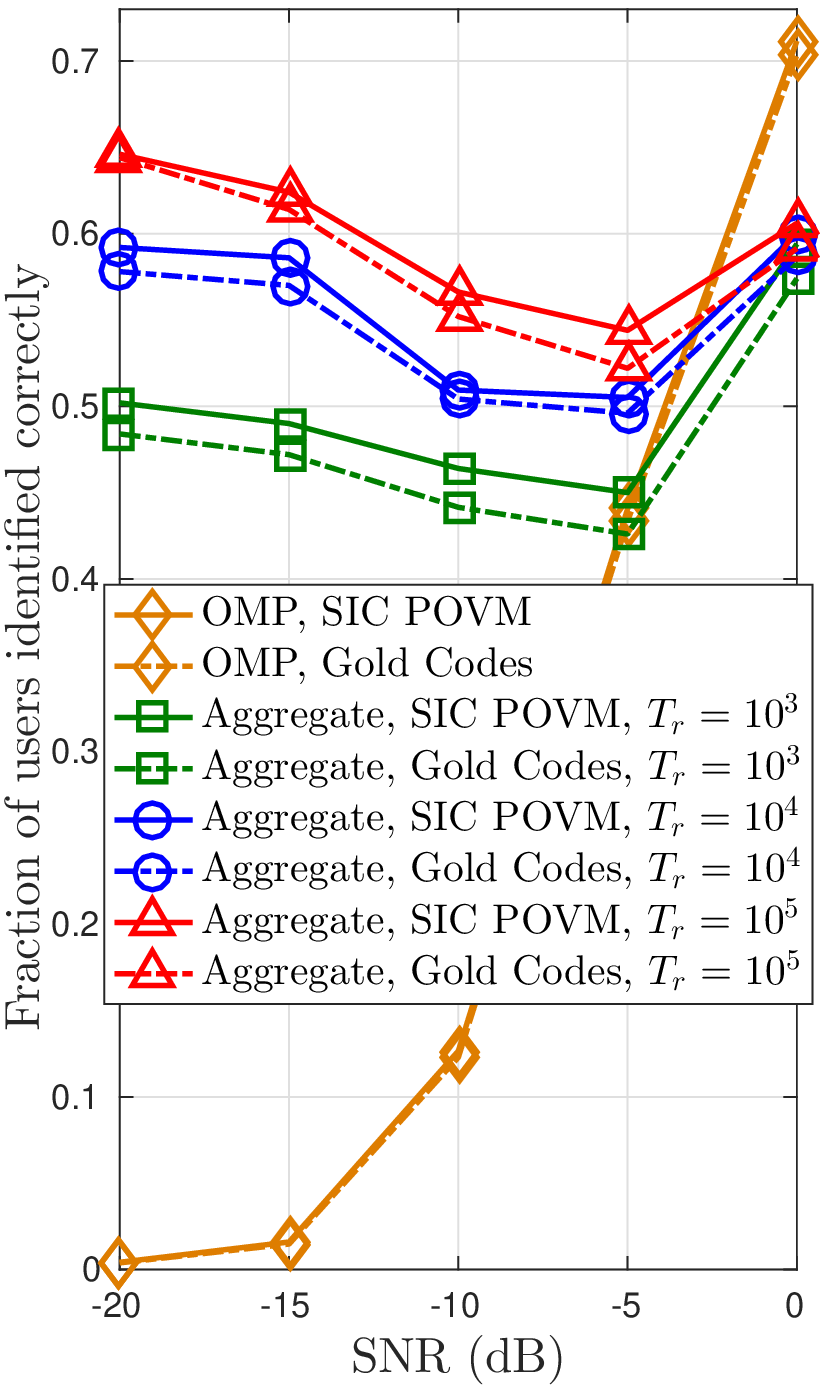}}
 \caption{Performance and user identification for massive random access with $P=1500, \Delta=8, K=5 $.}
    \label{fig:mra}
\end{figure}

\section{Conclusion}
In this paper, we address the  change detection problem with sparse signals  by combining the techniques from compressive sensing with asymptotically optimal CUSUM algorithm. We use the pdf-approximation and parameter-estimation based approaches when the support, signal variance and sparsity level of the signal are unknown.  Using deterministic sensing matrices with low mutual coherence further enhances the detection performance. We also analyze the detection performance of various decision statistics at different SNR levels. The problem of change detection when the non-zero entries of the sparse signal are correlated in time may be of interest for specific applications and may serve as a future scope of this work. Also, further research may be taken up to develop alternate techniques for detection when the distribution of the non-zero entries of the sparse signal, after the change point, is not known a priori.

\appendices
\section{Approximate pdf for Signal Energy}
\label{app:E_sigvar}
The covariance matrix of post change observation \eqref{Y_post} is $\vec{C}_\vec{y} = \vec{A}_{\mc{S}} \vec{C_x}\vec{A}_{\mc{S}}^* + \sigma_n^2 \vec{I}_M$. Let $\vec{\bar{U}}\vec{\Phi} \vec{\bar{U}}^*$ denote the eigen decomposition of $\vec{A}_{\mc{S}} \vec{C_x}\vec{A}_{\mc{S}}^*$. 
Note that, when $\vec{A}_{\mc{S}}$ is of rank $K$, then $\vec{\Phi}$ has $K$ real, non-zero eigenvalues. For convenience, let first $K$ entries $\{\phi_i, 1\leq i \leq K \}$ in the diagonal be non-zero.    
We define $\vec{z}[t]=\vec{\bar{U}}^* \vec{y}[t]$ and note that $\norm{\vec{z}[t]}_2^2= \norm{\vec{y}[t]}_2^2$, since $\vec{\bar{U}}$ is unitary. Now, $\vec{z}[t]$ is also a zero mean complex Gaussian random vector with covariance matrix 
$\vec{C}_\vec{z} = \vec{\Phi} + \sigma_n^2 \vec{I}_M$.
Since the off diagonal elements of $\vec{C}_\vec{z}$ are zero, the entries of $\vec{z}[t]$ are uncorrelated and hence, independent. Each entry of $\vec{z}[t]$ is distributed as
\begin{eqnarray*}
 z_i  \sim
    \begin{cases}
      \mathrm{CN}(0, \sigma^2_n + \phi_i ) , & 1 \leq i \leq K, \\
       \mathrm{CN}(0, \sigma^2_n), & i > K.
    \end{cases}
\end{eqnarray*}
 We have $\norm{\vec{z}}^2_2= \sum_{i=1 }^M |z_i|^2$, with each $|z_i|^2$ being distributed as
\begin{eqnarray*}
   |z_i|^2  \sim
    \begin{cases}
      \exp{\Big(\frac{1}{ \sigma^2_n+\phi_i }\Big)} , & 1 \leq i \leq K, \\
       \exp{\big(\frac{1}{\sigma^2_n}\big)}, & i > K.
    \end{cases}
\end{eqnarray*}
The Gramian matrices, $\vec{A}_\mc{S} \vec{C_x} \vec{A}_\mc{S}^* =\vec{A}_\mc{S} \vec{C_x}^{\frac{1}{2}}( \vec{A}_\mc{S} \vec{C_x}^{\frac{1}{2}})^*$ and $ (\vec{A}_\mc{S} \vec{C_x}^{\frac{1}{2}})^* \vec{A}_\mc{S}\vec{C_x}^{\frac{1}{2}}$ have the same $K$ non-zero eigenvalues. The entry in the $i$-th row and $\ell$-th column of  the $K \times K$ matrix $ (\vec{A}_\mc{S} \vec{C_x}^{\frac{1}{2}})^* \vec{A}_\mc{S}\vec{C_x}^{\frac{1}{2}}$ is $\langle \vec{a}_i,\vec{a}_\ell \rangle \sqrt{\sigma_{i}^2 \sigma_{\ell}^2}$, where $\vec{a}_i$ is the $i$-th column of $\vec{A}_{\mc{S}}$. 
Using the Gershgorin's circle theorem, we get the following bounds on each eigenvalue $\lbrace \phi_i: i=1 ,\dots,K \rbrace$,
 \begin{eqnarray}
 \label{eq:circle_theorem}
 \sigma_{i}^2 -  \alpha \sum\limits_{\substack{ \ell = 1,\cdots,K \\ \ell \neq i }} \sqrt{ \sigma_{i}^2 \sigma_{\ell}^2} \leq \phi_i \leq \sigma_{i}^2 + \alpha \sum\limits_{\substack{ \ell = 1,\cdots,K \\ \ell \neq i }}
\sqrt{ \sigma_{i}^2 \sigma_{\ell}^2}.
 \end{eqnarray}
Due to CLT, we use a Gaussian approximation for the post change pdf of $e[t]$, i.e., $\tilde{f}_1^E=\mathrm{N}(\mu_{E},\sigma^2_{E})$, where,
\begin{align}
\label{eq:E_mean}
\mu_E &=\E(E[t]) = \E(\norm{\vec{z}}_2^2)= \sum_{i=1}^M \E(|z_i|^2) \nonumber \\
 &= \sum_{i=1}^K  \E(|z_i|^2) + \sum\limits_{\substack{i > K }}  \E(|z_i|^2) \nonumber \\
 &= \sum_{i=1}^K  (\phi_i + \sigma^2_n) + (M-K)(\sigma_n^2) \nonumber \\
 &=\sum_{i=1}^K \phi_i + M\sigma_n^2.\\
\sigma_E^2 &= \text{Var}(E[t])=\text{Var}(\norm{\vec{z}}_2^2)= \sum_{i=1 }^M \text{Var}(|z_i|^2) \nonumber\\
 &= \sum_{i = 1}^K \text{Var}(|z_i|^2) + \sum_{i > K} \text{Var}(|z_i|^2) \nonumber\\
 &= \sum_{i = 1}^K  (\phi_i + \sigma^2_n)^2+(M-K)(\sigma_n^2)^2 \nonumber\\
\label{eq:E_var}
&= \sum_{i =1 }^K  \phi_i^2 +2\sigma_n^2\sum_{i=1}^K \phi_i+M(\sigma_n^2)^2.
\end{align} 
Next, we consider the following special cases:

Case 1: The unknown signal covariance matrix $\vec{C_x}$ is of the form $\diag([\sigma^2_1,\sigma_2^2,\cdots,\sigma_K^2])$. 
Since, $\sigma_{\min}^2 \leq \sigma_i^2 \leq \sigma_{\max}^2$, the worst case KL distance between post and pre-change pdf occurs when 
$\sigma_i^2 = \sigma_{\min}^2, 1\leq i\leq K$. Assuming this and using the bound in \eqref{circle_theorem}, we approximate the post change pdf with 
the smallest possible values for mean in \eqref{E_mean} and variance in \eqref{E_var}, 
which are given in \eqref{E_mean_approx} and \eqref{E_var_approx} respectively.  

Case 2: The signal covariance is of the form $\vec{C_x}=\sigma_x^2\vec{I}_K$. In \eqref{E_mean}, the sum of eigenvalues of the matrix $\sigma_x^2\vec{A}_\mc{S}\vec{A}_\mc{S}^*$ is $\sum_{i=1}^K\phi_i=\tr(\sigma_x^2\vec{A}_\mc{S}\vec{A}_\mc{S}^*) = \tr(\sigma_x^2\vec{A}_\mc{S}^*\vec{A}_\mc{S})=K\sigma_x^2$, which gives the mean in \eqref{E_post}. 
 Using \eqref{circle_theorem}, we substitute $\phi_i \geq  \phi_{\min} =\max\big\{0, \sigma_{x}^2 (1 - \alpha (K-1))\big \}$ for $1 \leq i\leq K$ in \eqref{E_var}, to compute the variance in \eqref{E_post} for worst case KL distance.

\section{Approximate Pdf for Correlator Statistics}
\label{app:G_sigvar}
After the change, $\vec{y}[t]= \sum_{{\ell}\in \mc{S} } \vec{a}_{\ell} x_{\ell} [t]  + \vec{n}[t]$  from \eqref{y_in_terms_supp}. Thus, $i$-th entry of $\vec{g}[t]$ is $g_i[t]= \langle \vec{a}_i,\vec{y}[t] \rangle = \langle \vec{a}_i, \sum_{{\ell} \in \mc{S} } \vec{a}_{\ell} x_{\ell}[t] + \vec{n}[t] \rangle$. Thus we have,
\begin{align*}
g_i[t]    &=\begin{cases}
     \langle  \vec{a}_i,\vec{n}[t] \rangle +  \sum\limits_{\substack{{\ell}\in \mc{S} \\ {\ell}\neq i}}\langle\vec{a}_i, \vec{a}_{\ell} \rangle x_{\ell} [t]+ \langle \vec{a}_i, \vec{a}_i\rangle x_i[t]  
       , \hspace{1mm}  \forall i \in \mc{S},\\
       \langle  \vec{a}_i,\vec{n}[t] \rangle + \sum\limits_{\substack{{\ell}\in \mc{S}}}\langle\vec{a}_i, \vec{a}_{\ell} \rangle x_{\ell}[t] , \hspace{22.5mm}   \forall i \notin \mc{S}.
    \end{cases}
 \end{align*}
For each of the above cases, the post-change pdf of  $g_i[t]$ can be approximated as
\begin{equation}
\label{eq:G_i}
 \tilde{f}_{1}^{G_i} =
    \begin{cases}
      \mathrm{CN}(0, \sigma_n^2  + \alpha ^2  \sigma_{\textrm{sum}}^2 + (1- \alpha^2)\sigma_{i}^2) , &  \forall i \in \mc{S},\\
       \mathrm{CN}(0,  \sigma_n^2 + \alpha ^2 \sigma_{\textrm{sum}}^2 , &  \forall i \notin \mc{S}.
    \end{cases}
 \end{equation}
where $\sigma_{\textrm{sum}}^2 = \sum_{\ell \in \mathcal{S}} \sigma_{\ell}^2$. The post-change pdf of the squared modulus of each entry $|g_i[t]|^2$ of $\vec{g}[t]$ is 
\begin{align}
\label{eq:C}
\tilde{f}_{1}^{|G_i|^2}  &=
    \begin{cases}
     \exp (\lambda_{\mc{S}}), &  \forall i \in \mc{S}, \\
      \exp(\lambda_{0}), &  \forall i \notin \mc{S},
    \end{cases} 
    \\   \text{where} \ \lambda_0 &= \frac{1}{\sigma_n^2 + \alpha ^2  \sigma_{\textrm{sum}}^2}, \lambda_\mc{S}= \frac{1}{\sigma_n^2  +  \alpha ^2  \sigma_{\textrm{sum}}^2  + (1- \alpha^2) \sigma_{i}^2} \nonumber.
\end{align}
Now, the maximum correlation is given by $c[t]= \underset{i=1,\dots, N}{\operatorname{max}} \lbrace |g_i[t]|^2 \rbrace$, which can be restated as
 \begin{align*} 
 c[t] &= \max{ \Big\lbrace \max{ \Big[ |g_i[t]|^2 \big\rvert_{i \in \mc{S}} \sim \exp{(\lambda_\mc{S})} \Big]} }, \\ 
 & \max{ \Big[|g_i[t]|^2 \big\rvert_{i \notin \mc{S}} \sim \exp{(\lambda_0)} \Big]} \Big \rbrace,
 \end{align*}
where $\lambda_{\mc{S}}$ and $\lambda_0$ are given in \eqref{C}.
Assuming the independence between the entries of $\vec{g}[t]$ (which is good approximation when the covariance between the entries is small), the post-change pdf of $c[t]$ is approximated as in \eqref{C_post}.

Consider the following special cases. 

Case 1: The unknown signal covariance matrix $\vec{C_x}$ is of the form $\diag([\sigma_1^2, \sigma_2^2, \cdots, \sigma_K^2])$. 
Again, the worst case KL distance between post and pre-change pdf occurs when $\sigma_i^2 = \sigma_{\min}^2,~\forall i$.
Now, substituting  $ \sigma_{\textrm{sum}}^2 = K\sigma_{\min}^2$ and $\sigma_i^2 = \sigma_{\min}^2$ in \eqref{G_i}, we get the approximate 
post-change pdf of $g_i[t]$, as given in \eqref{Gi_approx}. For $c[t]$, the post-change pdf is given by \eqref{C_post} and the modified values of rate parameters $\lambda_0$ and $\lambda_{\mc{S}}$ are obtained using similar approximations and are given by \eqref{C0_approx} and \eqref{Cs_approx}.

Case 2: The signal covariance matrix is of the form $\vec{C_x}=\sigma_x^2\vec{I}$. We substitute  $\sigma_{\textrm{sum}}^2 = K\sigma^2_x$ and $\sigma_i^2 =\sigma_x^2$ in \eqref{G_i} and \eqref{C} to obtain the approximate post-change pdfs of $g_i[t]$ and $c[t]$ respectively, as given in \eqref{G_post} and \eqref{C_post}.

\bibliographystyle{ieeetr}
\bibliography{macros_abbrev,references}

\end{document}